\documentclass[11pt,letter,citesort]{article}
\usepackage[margin=.9in]{geometry}
\usepackage[utf8]{inputenc}
\usepackage{amsmath,amssymb,amsfonts,amsthm}
\usepackage{graphicx}
\usepackage[table,pdftex,hyperref,svgnames]{xcolor}
\usepackage{bbold}
\usepackage{enumitem,xspace}
\usepackage{mathtools,bbm}
\usepackage[caption=false,font=footnotesize]{subfig}
\usepackage{algpseudocode,algorithm,algorithmicx}
\usepackage{cite}

        \usepackage{setspace}

\usepackage{hyperref,doi}
\hypersetup{colorlinks=true, linkcolor=blue, breaklinks=true, urlcolor=blue, citecolor=blue,}

\newcommand{\nSt}{\ensuremath{\textup{nSt}}\xspace}
\newcommand{\St}{\ensuremath{\textup{St}}\xspace}

\graphicspath{{./images/}{./fig}}
\setlist{nosep}

\newcommand{\R}{\mathbb{R}}

\def\la{\lambda}
\def\ve{\varepsilon}

\newcommand{\Ss}{\mathbb{S}}

\newcommand{\Qz}{\mathbb{R}^{n\times{n}}_{\textup{zd}}}
\newcommand{\Qzu}{\mathbb{S}^{n\times{n}}_{\textup{zd}}}
\newcommand{\Qzs}{\mathbb{R}^{n\times{n}}_{\textup{zd},\textup{symm}}}
\newcommand{\Qzus}{\mathbb{S}^{n\times{n}}_{\textup{zd},\textup{symm}}}
\newcommand{\Ps}{\mathbb{S}^{n\times{n}}_{\textup{zd},\textup{dss}}}
\newcommand{\Psa}{\mathbb{R}^{n\times{n}}_{\textup{zd},\textup{dss}}}

\renewcommand{\Qz}{\mathbb{R}^{n\times{n}}_{\textup{zero-diag}}}
\renewcommand{\Qzu}{\mathbb{S}^{n\times{n}}_{\textup{zero-diag}}}
\renewcommand{\Qzs}{\mathbb{R}^{n\times{n}}_{\textup{zero-diag},\textup{symm}}}
\renewcommand{\Qzus}{\mathbb{S}^{n\times{n}}_{\textup{zero-diag},\textup{symm}}}
\renewcommand{\Ps}{\mathbb{S}^{n\times{n}}_{\textup{zero-diag},\textup{dss}}}
\renewcommand{\Psa}{\mathbb{R}^{n\times{n}}_{\textup{zero-diag},\textup{dss}}}

\newcommand{\Rdiag}{\R^{n\times{n}}_{\textup{diag}}}
\newcommand{\jac}[1]{D\mkern-2mu{#1}}

\newtheorem{theorem}{Theorem}[section]
\newtheorem{corollary}[theorem]{Corollary}
\newtheorem{lemma}[theorem]{Lemma}

\newtheorem{conjecture}{Conjecture}

\newcommand{\until}[1]{\{1,\dots, #1\}}

\newcommand{\setdef}[2]{\{#1 \; | \; #2\}}

\newcommand{\map}[3]{#1: #2 \rightarrow #3}

\newcommand{\intersection}{\ensuremath{\operatorname{\cap}}}

\renewcommand{\natural}{{\mathbb{N}}}
\newcommand{\integer}{\ensuremath{\mathbb{Z}}}

\newcommand{\real}{\ensuremath{\mathbb{R}}}

\newcommand{\realnonnegative}{\ensuremath{\mathbb{R}}_{\ge 0}}

\newcommand\oprocendsymbol{\hbox{$\square$}}
\newcommand\oprocend{\relax\ifmmode\else\unskip\hfill\fi\oprocendsymbol}

\newcounter{saveenum}

\usepackage[prependcaption,colorinlistoftodos]{todonotes}

{
      \theoremstyle{plain}
      
}
\newtheorem{definition}{Definition}[section]

\newtheorem{remark}[theorem]{Remark}

\DeclareSymbolFont{bbold}{U}{bbold}{m}{n}
\DeclareSymbolFontAlphabet{\mathbbold}{bbold}
\newcommand{\vect}[1]{\mathbbold{#1}}

\newcommand{\sign}{\operatorname{sign}}
\newcommand{\diag}{\operatorname{diag}}
\newcommand{\grad}{\operatorname{grad}}

\newcommand{\trace}{\operatorname{trace}}

\newcommand\norm[1]{\left\lVert#1\right\rVert}

\newcommand{\D}{\mathcal{D}}
\newcommand{\Hs}{\mathcal{H}}

\newcommand{\Kulakowski}{Ku{\l}akowski\xspace}

\newcommand{\Fprod}[2]{\langle\!\langle#1,#2\rangle\!\rangle_F}

\newcommand{\Fnorm}[1]{\left\|#1\right\|_F}

\newcommand\blfootnote[1]{%
  \begingroup
  \renewcommand\thefootnote{}\footnote{#1}%
  \addtocounter{footnote}{-1}%
  \endgroup
}

\begin{document}

\title{Structural Balance via Gradient Flows\\ over Signed Graphs}

\author{Pedro Cisneros-Velarde, Noah E. Friedkin, Anton V. Proskurnikov, Francesco Bullo}

\maketitle

\begin{abstract}
  Structural balance is a classic property of signed graphs satisfying
  Heider's seminal axioms. Mathematical sociologists have studied balance
  theory since its inception in the 1940s.  Recent research has focused on
  the development of dynamic models explaining the emergence of structural
  balance.  In this paper, we introduce a novel class of parsimonious
  dynamic models for structural balance based on an interpersonal influence
  process.  Our proposed models are gradient flows of an energy function,
  called the dissonance function, which captures the cognitive dissonance
  arising from the violations of Heider’s axioms.  Thus, we build a new
  connection with the literature on energy landscape minimization.  This
  gradient-flow characterization allows us to study the transient and
  asymptotic behaviors of our model. We provide mathematical and numerical
  results describing the critical points of the dissonance function.
  \blfootnote{This work is supported by the U. S. Army Research Laboratory and the
  U. S. Army Research Office under grant number W911NF-15-1-0577. The
    views and conclusions contained in this document are those of the
    authors and should not be interpreted as representing the official
    policies, either expressed or implied, of the Army Research Laboratory
    or the U.S. Government.}
\blfootnote{Pedro Cisneros-Velarde (pacisne@gmail.com), Noah E.\ Friedkin and
    Francesco Bullo (\{friedkin,bullo\}@ucsb.edu) are with the Center for
    Control, Dynamical Systems and Computation, University of California,
    Santa Barbara.}

  \blfootnote{Anton V. Proskurnikov is with the Politecnico di Torino, Turin,
    Italy.}

\end{abstract}

\section{Introduction}

\subsubsection{Problem description and motivation}

Signed graphs represent networked systems with interactions classified as
positive or negative, e.g., cooperation or antagonism, promotion or
inhibition, attraction or repulsion. Such graphs naturally arise in diverse
fields, e.g., political science~\cite{MOJ-SN:15}, communication
studies~\cite{JL-DH-JK:10} and biology~\cite{CCL-CHL-CSF-HFJ-HCH:13}. In
sociology~\cite{NEF:98,DE-JK:10}, they are used to represent friendly or
antagonistic relationships, whereby signed edges may be interpreted as
interpersonal sentiment appraisals.  In the work by Heider~\cite{FH:46},
each individual appraises all other individuals either positively (friends,
allies) or negatively (enemies, rivals). Heider postulated four famous
axioms: (i) ``the friend of a friend is a friend,'' (ii) ``the enemy of a
friend is an enemy,'' (iii) ``the friend of an enemy is an enemy,'' and
(iv) ``the enemy of an enemy is a friend.''  Violations of these axioms
lead to cognitive tensions and dissonances that the individuals strive to
resolve; in this sense, Heider's axioms are consistent with the general
theory of cognitive dissonance~\cite{LF:1957}.  A signed network satisfying
Heider's axioms is called \emph{structurally balanced} \and {and} can have
only two possible configurations: either all of its members have positive
relationships with each other and become a unique faction, or there exist
two factions in which members of the same faction are friends but enemies
with every other member in the other faction. We refer
to~\cite{NEF:98,DE-JK:10} for textbook treatment and to~\cite{XZ-DZ-FYW:15}
for a recent comprehensive survey.

Whereas Heider's theory describes the qualitative emergence of structural
balance as the result of tension-resolving cognitive mechanisms, it does
not provide a quantitative description of these mechanisms and dynamic
models explaining the emergence of balance. The aim to fill this gap has
given rise to the important research area of \emph{dynamic structural
  balance}.  The \Kulakowski et al.~\cite{KK-PG-PG:05} model postulates an
influence process, whereby any individual $i$ updates her appraisal of
individual $j$ based on what others positively or negatively think about
$j$.  The Traag et al.~\cite{VAT-PVD-PDL:13} model postulates a homophily
process, whereby any individual $i$ updates her appraisal of $j$ according
to how much she agrees with $j$ on the appraisals of their common
acquaintances. Both models explain convergence to structural balance under
certain assumptions on the initial state (see below for more
information). Remarkably, both models assume the existence of so-called
\emph{self-appraisals} (loops in the signed graph) that strongly influence
the system dynamics.  Self-appraisals can be interpreted as individuals'
positive or negative opinions of themselves.

A second line of research, consistent with dissonance theory, has focused
on formulating social balance via appropriate energy functions. The
work~\cite{SAM-SHS-JMK:09} proposes an energy function for binary appraisal
matrices with global minima that represent structurally stable
configurations; it is argued that a dynamic structural balance model should
aim to navigate through this energy landscape and look for its minima. Some
models (e.g.,~\cite{TA-PLK-SR:05,TA-PLK-SR:06}) were designed precisely to
achieve this task. The work~\cite{GF-GI-CA:11} computes a distance to
balance via a combinatorial optimization problem, inspired by Ising models.


The purpose of this paper is threefold.  First, we aim to propose a more
parsimonious model of the influence process establishing structural
balance, that is, a model without self-appraisal weights. Our argument for
dropping these variables is that balance theory axioms do not include
self-appraisals, and the inclusion of such appraisals amounts to an
additional assumption and introduces unnecessary complexities.  Second, we
aim to connect the literature on dynamic structural balance with the
literature treating social balance as an optimization problem.  Finally, in
comparison with a known limitation of the \Kulakowski et al.~model, we aim
to emphasize through numerical simulations that our parsimonious model
predicts the emergence of structural balance also from asymmetric initial
configurations.

\subsubsection{Further comments on the state of the art}
We now present a summary of the current literature on dynamic structural
balance. Historically, the first models appeared in the physics
community~\cite{TA-PLK-SR:05,TA-PLK-SR:06,RF-DV-SY-OM:07}. These models
borrowed some concepts from statistical physics and had the
particularity of assuming that the appraisals between individuals are
binary valued (either $+1$ or $-1$).  At the same time, they rely on
hard-wired random mechanisms for the asynchronous updates of the
appraisals that lack a sociological insightful interpretation.

Another type of proposed models is based on discrete- and continuous-time
dynamical systems with real-valued appraisals. The seminal models of
  this kind are due to \Kulakowski et al.~\cite{KK-PG-PG:05} (later
analyzed more formally by~\cite{SAM-JK-RDK-SHS:11}) and Traag et
  al.~\cite{VAT-PVD-PDL:13}.  Models with real-valued appraisals
  capture not only signs, but also magnitudes of positive or negative
sentiments. All these models adopt synchronous updating and stipulate
sociological meaningful rules for the updating of appraisals, based on
either influence or homophily processes.  The following facts are known
about the \Kulakowski et al.~influence-based and the Traag et
al.~homophily-based models: the set of well-behaved initial conditions that
lead the social network towards social balance for the first model is a
subset of the set of normal matrices, while the second model can work under
generic initial conditions. Similar results are obtained by
\cite{WM-PCV-GC-NEF-FB:17f} for two discrete-time models based on influence
and homophily respectively: influence-based processes do not perform well
under generic initial conditions (in contrast to the homophily-based
processes).  Finally, only the models proposed in
\cite{WM-PCV-GC-NEF-FB:17f} and a variation of the model by \Kulakowski et
al.~proposed in the early work~\cite{KK-PG-PG:05}, have a bounded evolution
of appraisals, whereas the others have finite escape time.

Recent work has also started to focus on dynamic models for other relevant
configuration of signed graphs, e.g., configurations that satisfy only a
subset of the four Heider's axioms.
The work~\cite{NEF-AVP-FB:18m} provides a parsimonious model explaining the
emergence of a generalized version of structural balance from any initial
configuration; this model is based on an influence process of positive
contagion whereby influence is accorded only to positively-appraised
individuals.  A second model in this area is proposed
by~\cite{PJ-NEF-FB:13n}.  Finally, there has been a third type of models
that propose the emergence of structural balance or other generalized
balance structures for undirected graphs from a game theoretical
perspective~\cite{AvdR:11,MM-MF-PJK-HRR-MAS:11,PCV-FB:19g}.

\subsubsection{Contributions}

First of all, we contribute 
by proposing two new dynamic models that do not adopt the long-standing assumption of
self-appraisals and describe the evolution of signed networks without
self-loops. 
We argue that the introduction of self-weights
is poorly justified and that a model without them 
is a more faithful representation of Heider's theory.  The first model, called the
\emph{pure-influence model}, is a modification of the classic model by
\Kulakowski et al. which is obtained by eliminating self-appraisals (and
thus reducing the system's dimension). Analysis of its convergence
properties reduces to the analysis of our second model, which is called
the~\emph{projected pure-influence model} and which arises as a projection
of the first model onto the unit sphere. This second model has a
self-standing interest, since it enjoys bounded evolution of the
appraisals, while the first model shares the finite escape time property of
the classic model by~\Kulakowski et al.


Our second contribution is to build a bridge between dynamic structural
balance and balance as an optimization problem.  We propose an energy
function inspired by~\cite{SAM-SHS-JMK:09}, namely the \emph{dissonance
  function}, which measures the degree at which Heider's axioms are
violated among the individuals of a social network. We show that this
energy function has global minima that correspond to signed graphs
satisfying structural balance in the case of real-valued appraisals
(restricted on the unit sphere). Moreover, we show that our (projected)
pure-influence model is the gradient system of the dissonance function in
the case of undirected signed graphs, and hence the critical points of the
dissonance function are the equilibria of our dynamical system. Thus, we
establish a novel connection between dynamic structural balance and the
characterization of structural balance as the minima of an energy function
for real-valued appraisals. Remarkably, our derivations show that this
property of our models is enabled by the elimination of self-appraisals.
Thus, the models contributed in this paper may be considered as both an
interpersonal influence process and an extremum seeking dynamics for the
cognitive dissonance function.

Our third and more detailed contribution is the mathematical analysis of
the projected pure-influence model in the cases where the initial appraisal
matrix is symmetric.  In particular, we provide a complete characterization
of the critical points of the dissonance function (i.e., the equilibrium
points of the projected pure-influence model). This characterization relies
upon a special submanifold of the Stiefel manifold and its
properties. Along with the characterization of the critical points, we
analyze their local stability properties and provide some results on
convergence towards structural balance.

Our final contribution is a Monte Carlo numerical study of the
convergence of our models to structural balance under generic initial
conditions in both the symmetric and the asymmetric case. For the symmetric
case, our result is comparable to, but stronger than, what has already been
proved for the \Kulakowski et al.~model: our models converge to structural
balance under generic symmetric initial conditions. One key advantage of
our models, as compared with those by \Kulakowski et al., is that
convergence to structural balance emerges under generic asymmetric initial
conditions. Based on these numerical results, we formulate relevant
conjectures.


\subsubsection{Paper organization}
Section~\ref{sec:prelim} presents preliminary concepts.
Section~\ref{sec:models} presents our models and shows they are gradient
flows.  Section~\ref{sec:classification} and
Section~\ref{sec:convergence-analysis} contain an analysis of equilibria
and important convergence results, respectively.
Section~\ref{sec:simulations} contains numerical results and
conjectures. Finally, Section~\ref{sec:conclusion} contains some concluding
remarks.

\section{Preliminaries}
\label{sec:prelim}

\subsection{Signed weighted digraphs}

Given an $n\times n$ matrix $X=(x_{ij})$ with entries taking values in
$[-\infty,\infty]$, let $G(X)$ denote the signed directed graph where the
directed edge $i\xrightarrow[]{}j$ exists if and only if $x_{ij}\ne 0$, and
$x_{ij}$ represents its signed weight. The directed graph $G(X)$ is
complete if $X$ has no zero entries, except for the main diagonal.  $G(X)$
has no self-loops if and only if $X$ has zero diagonal entries.  Let
$x_{i*}$ denote the $i$th row of the matrix $X$ and $x_{*i}$ the $i$th
column of the matrix $X$.  Let $\sign(X)=(\sign (x_{ij}))$, where
$\map{\sign}{[-\infty,\infty]}{\{-1,0,+1\}}$ is as usual
\[
\sign(x)=
\begin{cases}
-1,\qquad\qquad&\text{if }x<0,\\
0,&\text{if }x=0,\\
+1,&\text{if }x>0.
\end{cases}
\]
Given a sequence $a_1,\ldots,a_n$, let $B=\diag(a_1,\ldots,a_n)$ denote the
diagonal $n\times n$ matrix $(b_{ij})$, where $b_{ii}=a_i$ and $b_{ij}=0$
for $i\ne j$. For an $n\times n$ matrix $X$, define
$\diag(X)=\diag(x_{11},\ldots,x_{nn})$.  For a vector $v\in\R^n$, define
$\diag(v)=\diag(v_1,\ldots,v_n)$.  Let $\vect{0}_n$ denote the $n\times{1}$
vector of zeros, and $\vect{0}_{n\times{n}}$ the $n\times{n}$ matrix with
zero entries.

Let $\succ$ and $\prec$ denote ``entry-wise greater than'' and ``entry-wise
less than,'' respectively. 

A \textit{triad} (if it exists) is a cycle between three nodes in
$G(X)$. The \textit{sign} of a triad is defined by the sign of the product
of the weights composing a triad. For example, the triad $i\to j\to k\to i$
has sign $\sign(x_{ij}x_{jk}x_{ki})$.

A real-valued matrix $Z$ is \emph{irreducible} if its graph $G(Z)$ is
strongly connected (a directed path between every two nodes exists) and
\emph{reducible} otherwise. If $Z$ is reducible, a permutation matrix $P$
exists such that the matrix 
\[
PZP^{\top}=
\begin{bmatrix}
Z_1 & *  & \ldots & *\\
0   & Z_2& \ldots & *\\
\vdots &&&\\
0 &&& Z_k
\end{bmatrix}
\]
is upper-triangular with irreducible blocks $Z_i$ (some of them can be $1\times 1$ matrices).
If $Z=Z^{\top}$, the latter matrix is block-diagonal matrix $PZP^\top=\diag(Z_1,\dots,Z_k)$ and the
graphs $G(Z_i)$ are the \emph{connected components} of the graph $G(Z)$.

\subsection{Sets of matrices and the Frobenius inner product}
Given two matrices $A,B\in\R^{n\times{n}}$, their Frobenius inner product
is defined by $\Fprod{A}{B}=\trace(B^\top A)$; the induced
norm 
is $\Fnorm{A}=\sqrt{\Fprod{A}{A}}$. Some
important properties for the trace operator are:
$\trace(A)=\trace(A^\top)$, $\trace(AB)=\trace(BA)$, and, for all
$d\in\natural$, $\trace(A^d)=\sum_{i=1}^n\lambda_i^d$ where 
$\lambda_1,\ldots,\lambda_n$ are the eigenvalues
of $A$.

Let $\Qz$ be the set of $n\times{n}$ real matrices with zero diagonal
entries, and $\Qzs$ be the set of symmetric matrices belonging to
$\Qz$. Let $\Ss^{n\times{n}}$ be the unit sphere in $\R^{n\times n}$, that is
$A\in\Ss^{n\times{n}}$ if and only if 
$A\in\R^{n\times{n}}$ with $\norm{A}_F=1$. Similarly, we define the sets $\Qzu=\Qz\cap\Ss^{n\times n}$ and
$\Qzus=\Qzs\cap\Ss^{n\times n}$.

Let $\Rdiag$ be the set of all real diagonal matrices and
$\R^{n\times{n}}_{\textup{sk-symm}}$ be the set of all skew-symmetric
matrices. Then, we have the following orthogonal decomposition of
$\R^{n\times{n}}$ equipped with the Frobenius inner product:
\begin{equation}
  \label{eq:F-decomposition}
  \R^{n\times{n}}=\R^{n\times{n}}_{\textup{sk-symm}} \oplus \Qzs \oplus \Rdiag.
\end{equation}

\subsection{A review on structural balance}

Throughout the paper we deal with social networks composed of $n\geq 3$
individuals, although the definition of structural balance
(Definition~\ref{def:sbal}) is formally applicable to the case of
degenerate networks with $n=1$ or $n=2$ nodes.

\begin{definition}[Appraisal matrix and network]
  We let the entry $x_{ij}$ of the matrix $X\in\R^{n\times{n}}$ denote the
  appraisal (or qualitative evaluation) held by individual $i$ of
  individual $j$. The sign of $x_{ij}$ indicates if the relationship is
  positive ($+1$), negative ($-1$) or of indifference ($0$). The magnitude
  of $x_{ij}$ indicates the strength of the relationship. $x_{ii}$ can be
  interpreted as $i$'s self-appraisal. We call $X$ the \emph{appraisal
    matrix}, and $G(X)$ the \emph{appraisal network}.
\end{definition}

\begin{definition}[Heider's axioms and social balance notions]
  \label{def:Heider+balance}
  The \emph{Heider's axioms} are
  \begin{enumerate}[label={H\arabic*)}]
  \item\label{H1} A friend of a friend is a friend,
  \item\label{H2} An enemy of a friend is an enemy,
  \item\label{H3} A friend of an enemy is an enemy,
  \item\label{H4} An enemy of an enemy is a friend.
  \end{enumerate}
  An appraisal network $G(X)$ is \emph{structurally balanced in Heider's
    sense}, if it is complete and satisfies axioms \ref{H1}-\ref{H4}.
\end{definition}

Consider a complete appraisal network $G(X)$. We call a \emph{faction} any
group of agents whose members positively appraise each other. We say two
factions are \emph{antagonistic} if every representative from one faction
negatively appraise every representative of the other faction.
It can be shown (\cite{FH:46,FH:53,DC-FH:56}) that Heider's structural
balance condition for $G(X)$ with $n\geq 3$ nodes holds if and only if
either the individuals constitute a single faction or can be partitioned
into two antagonistic factions.  The possession of the latter property may
thus be considered as an alternative definition of structural balance (and
is formally applicable to graphs without triads).

\begin{definition}[Structural balance]\label{def:structural-balance}
  \label{def:sbal}
  A complete appraisal network $G(X)$ is said to satisfy \emph{structural
    balance}, if $G(X)$ is composed by one faction or two antagonistic
  factions; or, whenever $n\geq 3$, equivalently, that all triads are
  positive, i.e., $x_{ij}x_{jk}x_{ki}>0$ for any different
  $i,j,k\in\until{n}$.
\end{definition}

Notice that a structurally balanced graph is always sign-symmetric: $\sign(x_{ij})=\sign(x_{ji})$ for any $i\ne j$. For simplicity we will say that a matrix $X$ \emph{corresponds to} structural balance whenever $G(X)$ satisfies structural balance. 

\section{Proposed models and representation as gradient flows}
\label{sec:models}

In this section we propose our models defining them over the set of
symmetric (appraisal) matrices, and the general setting will be postponed
until Section~\ref{sec:simulations} along some numerical results. Finally,
we prove that our models are gradient flows over a sociologically motivated
energy function.

\subsection{Pure-influence model}
\label{sec:a-new-model}
We propose our new dynamic model solely based on interpersonal appraisals.

\begin{definition}[Pure-influence model]
  The \emph{pure-influence model} is a system of differential equations on
  the set of zero-diagonal matrices $\Qz$ defined by
  \begin{equation}
    \label{inf-dyn}
    \dot{x}_{ij}=\sum_{\substack{k=1\\k\neq i,j}}^n x_{ik}x_{kj},
  \end{equation}
  for any $i,j\in\until{n}$ and $i\neq j$. Here $x_{ij}$, $i\neq j$, are
  the off-diagonal entries of a zero-diagonal matrix $X\in\Qz$. In
  equivalent matrix form, the previous equations read:
  \begin{equation}
    \label{inf-dyn-m}
    \dot{X}=X^2-\diag(X^2), \qquad X(0)\in\Qz.
  \end{equation}
\end{definition}

We interpret $X$ as the interpersonal appraisal matrix.  While
system~\eqref{inf-dyn} does not define the evolution of self-appraisals,
the matrix reformulation~\eqref{inf-dyn-m} ensures $\diag(\dot{X})=0$ and,
since $X(0)\in\Qz$ means $\diag(X(0))=\vect{0}_{n\times{n}}$, we have
$\diag(X(t))=\vect{0}_{n\times{n}}$ for all positive times $t$.

Our model is a modification of the classical model proposed by \Kulakowski
et al.~\cite{KK-PG-PG:05}, where self-appraisals play a crucial role in the
dynamics of the interpersonal appraisals. 

\begin{definition}[\Kulakowski et al. model]
  The \emph{\Kulakowski et al. model} is a system of differential equations
  on the state space $\R^{n\times{n}}$ defined by
\begin{subequations}     \label{inf-dynK}
  \begin{align}
    \label{inf-dynK1}
    \dot{x}_{ij}&=\sum_{k=1}^n x_{ik}x_{kj}=x_{ij}(x_{ii}+x_{jj})+\sum_{\substack{k=1\\k\neq i,j}}^n x_{ik}x_{kj},\\
    \dot{x}_{ii}&= x_{ii}^2 + \sum_{\substack{k=1\\k\neq i}}^n x_{ik}x_{ki}, \label{inf-dynK2}
  \end{align}
\end{subequations}
  for any $i\neq j\in\until{n}$. In equivalent matrix form, the previous
  equations read: $\dot{X}=X^2$.
\end{definition}

\begin{remark}[The problem with self-appraisals]
  The introduction of self-appraisals in model~\eqref{inf-dynK} is
  objectionable on several grounds.  The first conceptual problem is that
  self-appraisals are not considered in any definition of structural
  balance in the social sciences.  Heider's axioms in
  Definition~\ref{def:Heider+balance} do not take into account
  self-appraisals: social balance is a function of only interpersonal
  appraisals. Moreover, once self-appraisals are introduced, one needs to
  postulate why and how self-appraisals affect interpersonal appraisals,
  i.e., justify the choice of the first addendum for the right hand side
  of~\eqref{inf-dynK1}.  Finally, one needs to postulate how they evolve,
  i.e., justify the choice for the right hand side of~\eqref{inf-dynK2}.
  In summary, the pure influence model~\eqref{inf-dyn} avoids these
  difficulties and stays closer to the foundations of structural balance,
  in which individuals are attending only to interpersonal appraisals.
  Even though $\dot{X}=X^2$ may appear mathematically simpler or more
  elegant than $\dot{X}=X^2-\diag(X^2)$, we believe the latter model is
  actually more parsimonious, lower dimensional, and more faithful to
  Heiders' axioms.
\end{remark}

One easily notices the following important property of the pure-influence
model~\eqref{inf-dyn-m}: the right-hand side is an analytic function of $X$
so that the equation enjoys (local) existence and uniqueness of the
solutions. A second property is that, if $X(0)=X(0)^{\top}$, then
$X(t)=X(t)^{\top}$ for all subsequent times. This implies that the
pure-influence model is well defined over the set of symmetric (zero
diagonal) matrices $\Qzs$.

\subsection{Dissonance function}

We introduce and study the properties of a useful \emph{dissonance
  function} that summarize the total amount of cognitive
dissonances~\cite{LF:1957} among the members of a social network due to the
lack of satisfaction of Heider's axioms.  Recall that, according to
Definition~\ref{def:structural-balance}, a triad $i\to j\to k\to i$
satisfies the axioms if and only if $x_{ij}x_{jk}x_{ki}>0$.


\begin{definition}[Dissonance function]
  The \emph{dissonance function} $\map{\D}{\Qz}{\R}$ is
  \begin{equation}
    \label{beta_hat1}
    \D(X) = - \!\!\!\!\!\! \sum_{\substack{i,j,k=1\\ i\neq{j},j\neq{k},k\neq{i} }}^n \!\!\!\!\!\!
    x_{ij}x_{jk}x_{ki}= - \trace(X^3)=-\sum_{i=1}^n\lambda_i^3,
  \end{equation}
  where $\{\lambda_i\}_{i=1}^n$ is the set of eigenvalues of $X$.
\end{definition}
We plot $\D$ in a low-dimensional setting in
Figure~\ref{fig:dissonancefunction}.

\begin{figure}[ht]
  \centering
    \includegraphics[width=0.4\textwidth]{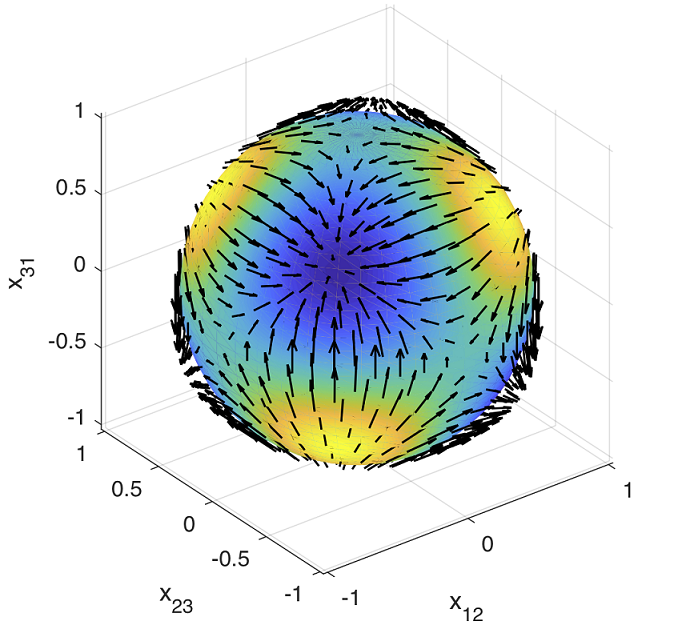}
    \includegraphics[width=0.48\textwidth]{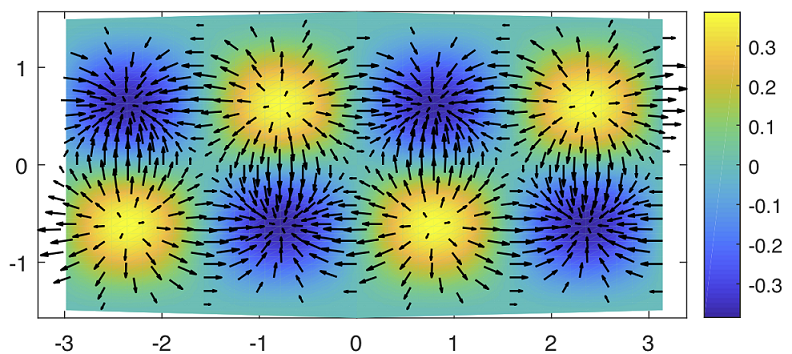}
    \caption{For $n=3$, an arbitrary symmetric unit-norm zero-diagonal
      matrix $X\in\Qzus$ is described by $(x_{12},x_{23},x_{31})$ with
      these coordinates living in the sphere
      $x_{12}^2+x_{23}^2+x_{31}^2=1$. In the upper figure, we plot this
      sphere with a heatmap, with dark blue being the lowest value and
      light yellow the largest value, according to the evaluation of the
      dissonance function $\mathcal{D}(X)$. The function has four global
      minima corresponding to the four possible configurations of $G(X)$
      satisfying structural balance, and we can qualitatively appreciate
      the convergence of solution trajectories to these minima in the
      superimposed vector field on the sphere. The lower figure is a
      stereographic projection of the upper figure.}
    \label{fig:dissonancefunction}
\end{figure}


Energy landscapes in social balance theory are studied
in~\cite{SAM-SHS-JMK:09,GF-GI-CA:11}. 
Our proposed dissonance function is the extension to $\Qz$ of the energy
function proposed by~\cite{SAM-SHS-JMK:09} for the setting of binary-valued
symmetric appraisal matrices.
For binary-valued appraisals, the global minima of $\D$ correspond to
networks that satisfy structural balance, since all triads are positive
(see Definition~\ref{def:sbal}). Thus, $\D$ naturally measures to which
extent Heider's axioms are violated in a (complete) social network.

\begin{lemma}[Properties of the dissonance function]
  \label{lem-DH}
  Consider the dissonance function $\D$ and pick $X\in\Qz$. Then
  \begin{enumerate}[label=(\roman*)]
  \item $\D$ is analytic and attains its maximum and minimum values on any
    compact matrix subset of $\Qz$,
  \item if $G(X)$ satisfies structural balance, then $\D(X)<0$, 
  \item $\D(X)=\D(X^\top)$, 
  \item $\D(X) = -\Fprod{X^2}{X^\top}$ 
  \setcounter{saveenum}{\value{enumi}}
\end{enumerate}
Additionally, if $\norm{X}_F=1$, that is, $X\in\Qzu$, then
\begin{enumerate}[label=(\roman*)]\setcounter{enumi}{\value{saveenum}}
\item $-1 \leq \D(X) \leq 1$. \label{lastin}
\end{enumerate}
\end{lemma}
\begin{proof}
  Here we show only property~\ref{lastin}, since the other properties are
  easily verified from the definition of $\D$. The key step is to show that
  $\Fnorm{X}\leq1$ implies $\Fnorm{X^2}\leq1$. The Cauchy-Schwartz
  inequality leads to:
  \begin{align*}
    \Fnorm{X^2}^2 &= \sum_{i,j=1}^n (X^2)^2_{ij} = \sum_{i,j=1}^n (X_{i*} X_{*j})^2 \\
    &\leq \sum_{i,j=1}^n \|X_{i*}\|_2^2 \|X_{*j}\|_2^2= \Big(\sum_{i=1}^n \|X_{i*}\|_2^2\Big) \Big(   \sum_{j=1}^n\|X_{*j}\|_2^2 \Big)\\
    &= \Big(\sum\nolimits_{i,k=1}^n x_{ik}^2\Big)^2 = \Fnorm{X}^2 =1.
  \end{align*}
  Since $\D$ is a Frobenius inner product of vectors with at-most unit
  norm, it is bounded by $1$ in absolute value.
\end{proof}


\subsection{Transcription on the unit sphere and the projected pure-influence model}

We start by noting a simple fact.  Given a trajectory
$\map{X}{\realnonnegative}{\Qz}$ with $X(t)\neq \vect{0}_{n\times{n}}$ for
all $t$, there exist unique trajectories
$\map{\eta}{\realnonnegative}{\realnonnegative}$ and
$\map{Z}{\realnonnegative}{\Qzu}$ such that $X(t) = \eta(t) Z(t)$, where
$\eta(t)=\Fnorm{X(t)}$ and $Z(t)=X(t)/\Fnorm{X(t)}$.

\begin{theorem}[Transcription of the pure-influence model]
  \label{eg-trans}
 The pure-influence model~\eqref{inf-dyn} can be expressed as the
    following system of differential equations:
    \begin{subequations} \label{eq1o}
      \begin{align}
        \dot{Z}&=\eta\mathcal{P}_{Z^\perp}(Z^2-\diag(Z^2)) \nonumber \\
        &=\eta(Z^2-\diag(Z^2)+\D(Z)Z),       \label{eq1o-1} \\
        \dot{\eta}&=-\D(Z)\eta^2, \label{eq1o-2}
      \end{align}
    \end{subequations}
    where $\map{\eta}{\realnonnegative}{\realnonnegative}$ and
    $\map{Z}{\realnonnegative}{\Qzus}$.  Here $\mathcal{P}_{Z^\perp}$ is
    the orthogonal projection onto $\operatorname{span}\{Z\}^\perp$ in the
    vector space of square matrices with the Frobenious inner product.
\end{theorem}

\begin{proof}
We start by computing
$\dot{X}=\dot{\eta}Z+\eta\dot{Z}=X^2-\diag(X^2)$.  Since
$X^2-\diag(X^2)=\eta^2\left(Z^2-\diag(Z^2)\right)$, we know
\begin{equation}
  \label{one2}
  \dot{\eta}Z+\eta\dot{Z}=\eta^2\left(Z^2-\diag(Z^2)\right).
\end{equation}
Recall that $\Fnorm{Z(t)}=1$ implies $\Fprod{Z(t)}{\dot{Z}(t)}=0$, that is,
$Z(t)\perp\dot{Z}(t)$.  Computing the Frobenius inner product with $Z(t)$
on both sides of~\eqref{one2}, we obtain
\begin{equation*}
  \begin{split}
  \dot{\eta}&=\eta^2\Fprod{Z(t)}{Z^2(t)-\diag(Z^2(t))}
  =\eta^2\Fprod{Z(t)}{Z^2(t)}\\
  &= -\D(Z(t))\eta^2.	\end{split}
\end{equation*}
where we have used the
decomposition~\eqref{eq:F-decomposition}. Substituting this equation into
equation~\eqref{one2}, one arrives at
$\dot{Z}=\eta\left(Z^2-\diag(Z^2)+\D(Z)\right)$.

Given $Y\in\R^{n\times{n}}$, let $\mathcal{P}_Z(Y)=\langle Y,Z\rangle Z$,
i.e., $\mathcal{P}_Z$ is the orthogonal projection operator onto the linear
space spanned by $Z$; and let
$\mathcal{P}_{Z^\perp}(Y)=Y-\mathcal{P}_Z(Y)=Y-\langle Y,Z\rangle Z$ be the
orthogonal projection onto the space perpendicular to the linear space
spanned by $Z$.  Then, we observe that $\mathcal{P}_{Z^\perp}(Z)=0$ and
$\mathcal{P}_{Z^\perp}(\dot{Z})=\dot{Z}$. Using these results, we apply
$\mathcal{P}_{Z^\perp}$ to both sides of~\eqref{one2} and obtain
$\dot{Z}=\eta\mathcal{P}_{Z^\perp}(Z^2-\diag(Z^2))$.  This concludes the
proof of equations~\eqref{eq1o}.
\end{proof}

In what follows, we are primarily interested in the
dynamics~\eqref{eq1o-1}, describing the behavior of the bounded component
$Z(t)$. For our needs, it is convenient to change the time variable
(Lemma~\ref{thaux2}) by getting rid of $\eta$ and replacing~\eqref{eq1o} by
the following dynamical system on the unit sphere.

\begin{definition}[Projected pure-influence model]
The \emph{projected pure-influence model} is a system of differential equations on the manifold $\Qzu$ defined by
\begin{equation}
  \label{def:projected-pure-influence}
  \dot Z=Z^2-\diag(Z^2)+\D(Z)Z.
\end{equation}
\end{definition}

Similarly, projecting onto the unit sphere leads to a new model based on
the \Kulakowski et al. model.

\begin{definition}[Projected \Kulakowski et al. model]
  The \emph{projected \Kulakowski et al. model} is a system of differential
  equations on the manifold of symmetric unit-Frobenius norm matrices
  matrices defined by
\begin{equation}
  \label{def:projected-pure-influenceK}
  \dot Z(t)=Z^2+\D(Z)Z.
\end{equation}
\end{definition}

\subsection{Pure-influence is the gradient flow of the dissonance function}

In this section we let $\grad\D$ denote the gradient vector field on $\Qz$
defined by the dissonance function $\D$. We also let $\D\big|_{\Qzus}$
denote the restriction of $\D$ onto the set $\Qzus$.  With this notation,
we now present the first of our main results.

\begin{theorem}[The pure-influence models over symmetric matrices are gradient flows]
  \label{th-mini-gr}
  Consider the pure-influence model~\eqref{inf-dyn} with $X(0)\in\Qzs$ and
  the projected pure-influence model~\eqref{def:projected-pure-influence}
  with $Z(0)\in\Qzus$. Then
  \begin{enumerate}
  \item \label{lasym22}$t\mapsto X(t)$ remains in the set $\Qzs$ and
    \begin{equation}
      \label{eq:pureinf=gradient}
      \dot{X}=-\tfrac{1}{3}\grad \D(X),
    \end{equation}
  \item\label{lasym0} $t\mapsto Z(t)$ remains in the set $\Qzus$ and
    \begin{equation}
      \label{x1}
      \dot{Z} = -\tfrac{1}{3}  \mathcal{P}_{Z^\perp}\!\big( \grad \D(Z) \big)
      = -\tfrac{1}{3} \grad\D\Big|_{\Qzus}\!\!\!\!\!\!\!\!\!(Z).
    \end{equation}
  \end{enumerate}
\end{theorem}

In other words, the projected pure-influence
model~\eqref{def:projected-pure-influence} is, modulo a constant factor, 
the gradient flow of the dissonance function $\D$ restricted to the manifold of
zero-diagonal unit-norm symmetric matrices $\Qzus$.

\begin{proof}[Proof of Theorem~\ref{th-mini-gr}]
  The forward invariance of the set of symmetric matrices in both statements
  is immediate.  To prove equation~\eqref{x1}, we adopt the slight abuse of
  notation $$\grad\D(Z) = \grad\D\Big|_{\Qzus}\!\!\!\!\!\!\!\!\!(Z).$$  With this
  notation, note that $Z\mapsto\grad\D(Z)$ is the unique vector field on
  $\Qzus$ satisfying, along any differentiable trajectory $t\mapsto{Z(t)}$,
  \begin{equation}\label{eq:fb:blast}
    \frac{d}{dt}\D(Z(t)) = \Fprod{\grad\D(Z(t))}{\dot{Z}(t)}.
  \end{equation}
  Note that, here, both $\grad\D(Z(t))$ and $\dot{Z}(t)$ take value on the
  tangent space to the manifold $\Qzus$.

  Now, using the various properties of the trace inner product (e.g.,
  $\dot{Z}(t)\perp Z(t)$), we compute
  \begin{equation*}
    \begin{split}
      \dot{\D}(Z(t)) &=
      -(\trace(\dot{Z}(t)Z(t)Z(t))+\trace(Z(t)\dot{Z}(t)Z(t)))\\
      &\phantom{ = }\quad+\trace(Z(t)Z(t)\dot{Z}(t)) \\
      &=-3\trace(\dot{Z}(t)Z^2(t)) = -3\Fprod{\dot{Z}(t)}{Z^2(t)} \\
      &= -3\Fprod{\dot{Z}(t)}{Z^2(t)-\diag Z^2(t)+\D(Z(t))Z(t)}.
    \end{split}
  \end{equation*}
  Recalling that $Z^2-\diag{Z^2}+\D(Z)Z \overset{\eqref{eq1o-1}}{=}
  P_{Z^{\top}}(Z^2-\diag Z^2)$ belongs to the tangent space to the manifold
  $\Qzus$ at the point $Z(t)$, one arrives at the equality
  \begin{equation*}
    \grad\D(Z) = -3 \Big( Z^2-\diag Z^2+\D(Z)Z \Big).
  \end{equation*}
  This concludes the proof of statement~\ref{lasym0}.  Finally,
  equation~\eqref{eq:pureinf=gradient} can be proved in a similar way.
\end{proof}

\section{Classification of symmetric equilibria}
\label{sec:classification}
In this section we give the complete classification of the symmetric
equilibria in the projected pure-influence
model~\eqref{def:projected-pure-influence}; the classification of general
asymmetric equilibria remains an open problem.  Thanks to
Theorem~\ref{th-mini-gr}, all symmetric equilibria of the projected
pure-influence model are critical points of the dissonance function
$\D$. It is useful to write the equilibrium equation:
\begin{equation}
  \label{eq.equil}
  Z^2+\D(Z)Z-\diag(Z^2)=0, \quad Z\in\Qzus.
\end{equation}
Note that the equilibria $Z^*$ with $\D(Z^*)=0$ correspond to equilibria of
the original system~\eqref{inf-dyn-m} $X(t)\equiv X^*=\eta(0)Z^*$, whereas
the others with $\Hs(Z^*)\neq 0$ lead to
\begin{equation*}
 X(t)=\eta(t)Z^*, \quad\eta(t)=\frac{\eta(0)}{1+t\eta(0)\D(Z^*)}
\end{equation*}
defined for $t\in[0,\frac{1}{\eta(0)\D(Z^*)})$ if $\D(Z^*)<0$ (for which
  the solution is unbounded) or for $t\geq 0$ if $\D(Z^*)>0$.

\subsection{Normalized Stiefel matrices}

To start with, we introduce a special important manifold of non-square
matrices that we will use throughout the paper.

\begin{definition}[Normalized Stiefel matrices]\label{def-unit}
  A matrix $V\in\real^{n\times k}$, for $k\le n$, is \emph{normalized
    Stiefel} (\nSt), if
  \begin{enumerate}
  \item\label{eq.gu1} the columns of $V$ are pairwise orthogonal unit
    vectors, i.e., $V^{\top}V=I_{k}$;
    
  \item\label{eq.gu2} the norm of each row is the same (obviously, it must
    be $\sqrt{k/n}\le 1$): $\diag(VV^{\top})=n^{-1}k I_n$.   
  \end{enumerate}
  Let $\nSt(n,k)\subseteq\R^{n\times k}$ denote the set of normalized
  Stiefel matrices.
\end{definition}

In general, the rows of an \nSt matrix \emph{need not} be orthogonal. We
recall from~\cite{IMJ-NJH:76} the notion of \emph{compact Stiefel
  manifold}, denoted by $\St(k,n)=\setdef{X\in\R^{n\times{k}}}{X^\top X =
  I_k}$.

\begin{lemma}[Characterization of \nSt matrices]\label{ch-gunit}
  The set $\nSt(n,k)$, $k\leq n$, is a compact and analytic submanifold of
  $\R^{n\times k}$ of dimension $(k-1)n+1-k(k+1)/2$, and it is also a
  submanifold of the compact Stiefel manifold~(and thus,
  $\nSt(n,k)\subseteq \St(k,n)$). Moreover,
\begin{enumerate}[label=(\roman*)]
\item \label{st-1}$\nSt(n,n)$ is the set of orthogonal matrices,
\item \label{st-2}for $k=1$, the matrix $V$ is \nSt if and only if
\begin{equation}\label{eq.k1-general}
  V=\frac{1}{\sqrt{n}}\begin{bmatrix}s_1\\ \vdots \\ s_n
  \end{bmatrix},
\end{equation}
for any numbers $s_i\in\{-1,+1\}$, $i\in\until{n}$,
\item \label{st-3}for $k=2$, the matrix $V$ is \nSt if and only if
\begin{equation}\label{eq.k2-general}
  V= \sqrt{\frac{2}{n}}
\begin{bmatrix}
  \cos\alpha_1 & \sin\alpha_1\\
  \vdots & \vdots\\
  \cos\alpha_n & \sin\alpha_n
\end{bmatrix},
\end{equation}
for any set of angles $\alpha_1,\dots,\alpha_n$ satisfying
\begin{equation}
  \label{eq.ngon}
  \sum_{m=1}^n e^{2\alpha_m \sqrt{-1}}=0.
\end{equation}
\end{enumerate}
\end{lemma}

We postpone the proof of Lemma~\ref{ch-gunit} to
Appendix~\ref{sec:proofs}. We remark that in the case of $n=k=2$, the
constraint~\eqref{eq.ngon} implies that $2\alpha_2=\pi+2\pi s+2\alpha_1$,
where $s\in\integer$, that is, $\alpha_2=\pi/2+\pi s+\alpha_1$ and
$\cos\alpha_2=(-1)^{s+1}\sin\alpha_1$,
$\sin\alpha_2=(-1)^{s}\cos\alpha_1$. Thus, the matrices in $\nSt(2,2)$ are
orthogonal $2\times 2$ matrices (representing rotations or rotations with
reflection):
\[
V=\begin{bmatrix}
\cos\alpha_1 & \sin\alpha_1\\
-\ve\sin\alpha_1 & \ve\cos\alpha_1
\end{bmatrix},\quad \ve\in\{-1,+1\}.
\]
For a general $k$, it is difficult to give a closed-form description of all
matrices from $\nSt(n,k)$. However, there are simple examples of matrices
from $\nSt(n,k)$ in the case where $n=2k$, including every matrix of the form
\[
V=\frac{1}{\sqrt{2}}
\begin{bmatrix}
U_1\\
U_2
\end{bmatrix},
\]
where $U_i$ are orthogonal $k\times k$ matrices.


\subsection{Technical results}

We here present two technical results proved in Appendix~\ref{sec:proofs}.

\begin{lemma}\label{lem.tech2}
Suppose that $Z^2-2\alpha Z=\beta I_n$ for some symmetric $n\times n$ matrix $Z$ with $\diag Z=0$. Then $Z$ can be decomposed as
\begin{equation}\label{eq.special1}
Z=pVV^{\top}-qI_n=Z^{\top}
\end{equation}
for some $V\in \nSt(n,k)$ ($1\le k<n$) and constants $p,q\ge 0$ such that $pk=qn$, $2\alpha=\theta=p-2q$ and $\beta=q(p-q)$. Namely,
$p=2\sqrt{\alpha^2+\beta},\quad q=\sqrt{\alpha^2+\beta}-\alpha$.
\end{lemma}

\begin{corollary}\label{cor.diag}
Given a matrix $Z=Z^{\top}$ with $\diag(Z)=0$, the matrix $Z^2-2\alpha Z$ is
diagonal with $s$ different eigenvalues $\beta_1<\ldots<\beta_s$ of
multiplicities $n_1,\ldots,n_s$ respectively ($n_1+n_2+\ldots+n_s=n$) if
and only if there exists such a permutation matrix $S$ that
\[
SZS^{-1}=\diag(Z_1,\ldots,Z_s),
\]
where each $Z_i$ is decomposed as~\eqref{eq.special1} with parameters $p_i,q_i,V_i$, where $V_i\in \nSt(n_i,k_i)$ for some $k_i<n_i$ and
\begin{equation}\label{eq.pq-general1}
p_i=2\sqrt{\alpha^2+\beta_i},\quad q_i=\sqrt{\alpha^2+\beta_i}-\alpha.
\end{equation}
Thus, for \emph{irreducible} $Z=Z^{\top}$ the matrix $Z^2-2\alpha Z$ is diagonal if and only if $Z$ is decomposed as~\eqref{eq.special1} with $p,q\ge 0$.
\end{corollary}
\color{black}

\subsection{Classification of irreducible symmetric equilibria}

\begin{theorem}[Irreducible equilibria for the projected pure-influence model]\label{cor.tech}
  For the projected pure-influence model~\eqref{def:projected-pure-influence},
  \begin{enumerate}
  \item \label{1a0} all irreducible symmetric equilibria are of the form
    \begin{equation}
      \label{eq.special}
      Z^*=pVV^{\top}-qI_n,
    \end{equation}
    with $V\in \nSt(n,k)$, $k<n$, and
    \begin{equation}\label{eq.pq}
      p=\sqrt{\frac{n}{k(n-k)}},\quad q=\sqrt{\frac{k}{n(n-k)}};
    \end{equation}
\item \label{la1}$Z^*$ has $k$ positive eigenvalues with value $p-q$ and $n-k$ negative eigenvalues with value $-q$;
\item \label{la2}the dissonance function  satisfies
  \begin{equation}\label{eqH}
    \D(Z^*)=-\frac{n-2k}{\sqrt{kn(n-k)}},
  \end{equation}
  and the right-hand side is monotonically increasing in $k\in\until{n-1}$
  (see Figure~\ref{cos_D_eval}).
  \end{enumerate}
\end{theorem}
\begin{figure}[ht!]
  \centering
  \includegraphics[width=0.6\linewidth]{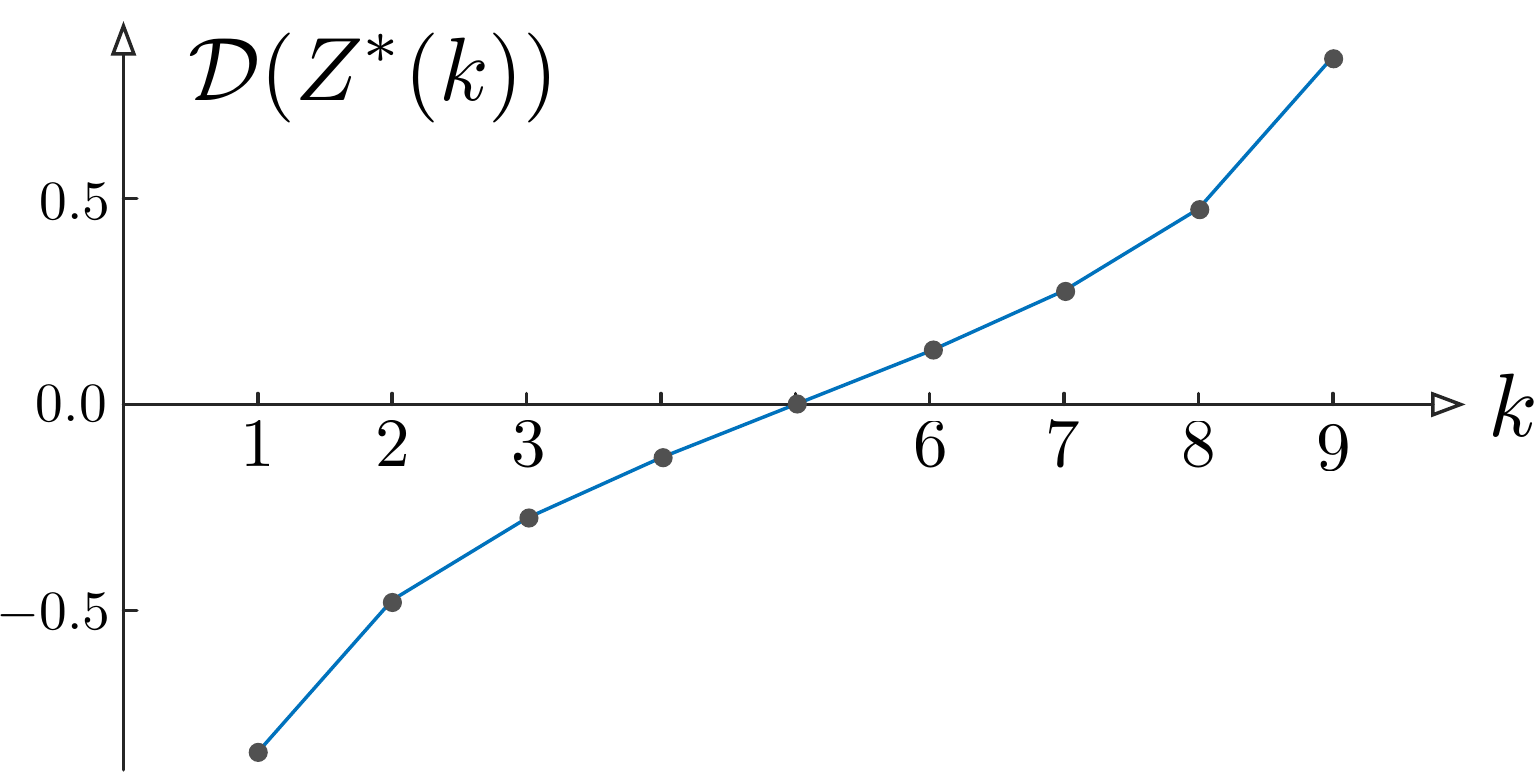}
  \caption{For a network with size $n=10$, the dissonance function $\D$
    evaluated on all irreducible symmetric equilibria with $k\in\until{9}$
    positive eigenvalues, according to equation~\eqref{eqH}.}
  \label{cos_D_eval}
\end{figure}

\begin{proof}
We start by proving a technical statement. Pick $V\in \nSt(n,k)$, $p,q$ real
numbers and $\theta=p-2q$. Then, 
the matrix
$Z=pVV^{\top}-qI_n=Z^{\top}$ satisfies the following properties:
\begin{enumerate}[label=(\alph*)]
\item $Z^2-\theta Z=q(p-q)I_n$, and thus $\diag(Z^2)=\theta\diag(Z)+q(p-q)I_n$;\label{in1}
\item for any $p\ne 0$, the matrix $Z$ has two eigenvalues $p-q$ and $(-q)$
  whose multiplicities are $k$ and $(n-k)$ respectively;\label{in2}
\item the eigenspaces corresponding to $p-q$ and $-q$ are the image of $V$ and the kernel of $V^{\top}$ respectively; \label{in3}
\item $\diag(Z)=\vect{0}_{n\times{n}}$ if and only if $pk=qn$;
  in this situation, $\trace(Z^2)=q(p-q)n$ and $\D(Z)=-\trace(Z^2Z^{\top})=-\theta nq(p-q)$.\label{in4}
\end{enumerate}
\noindent To prove~\ref{in1}, recall that $V^{\top}V=I_k$ and therefore
\begin{align*}
Z^2&=p^2VV^{\top}VV^{\top}+q^2I_n-2pqVV^{\top}=p\theta VV^{\top}+q^2I_n\\
&=\theta Z+(pq-q^2)I_n.
\end{align*}
To prove~\ref{in2} and~\ref{in3}, notice that for any vector $z=Vy$ one has
$VV^{\top}z=V(V^{\top}V)y=Vy=z$, and thus $Zz=(p-q)z$. The space of such
vectors is nothing else than the image of $V$ and has dimension $k$ (recall
that the columns of $V$ are orthogonal, and hence are linearly
independent). If $V^{\top}z=0$, then $Zz=-qz$, and the dimension of
$\ker(V^{\top})$ is $(n-k)$. Since $Z=Z^{\top}$ and $p-q\ne -q$ (except for
the case where $p=q=0$ and $Z=0$, which is trivial), the two eigenspaces
are orthogonal and their sum coincides with $\R^n$. Hence, there are no
other eigenvalues.  To prove~\ref{in4}, note first $p\diag
(VV^{\top})=(pk/n) I_n$, and thus $\diag(Z)=0$ if and only if
$pk/n=q$. Using statement~\ref{in1}, one shows that in this situation
$\diag(Z^2)=q(p-q)I_n$ and hence $\trace(Z^2)=q(p-q)n$. Thanks
to~\ref{in1}, $Z^3=\theta Z^2+q(p-q)Z\Longrightarrow
\trace(Z^3)=\theta\trace(Z^2)=\theta nq(p-q)$, which finishes the proof
of~\ref{in4}.

Now, to prove the statement~\ref{1a0} of the theorem, note first that
from~\ref{in1} and equation~\eqref{eq.equil}, it follows from
Corollary~\ref{cor.diag} that every irreducible equilibrium is decomposed
as~\eqref{eq.special} with some $p,q\geq 0$.  Moreover, note that
from~\ref{in1} and~\ref{in4}, it also follows that
equation~\eqref{eq.equil} holds if and only if $pk=qn$ (which comes from
$Z$ having zero diagonal entries and so $k<n$) and $pq-q^2=1/n$ (which
comes from $\trace(Z^2)=1$). This implies that $q=\sqrt{\frac{k}{n(n-k)}}$
and $p=\sqrt{\frac{n}{k(n-k)}}$.

Finally, statement~\ref{la1} follows from~\ref{in2}; and~\ref{la2} is
obtained by substituting the values of $p$ and $q$ to the definition of the
dissonance function~\eqref{beta_hat1} and noting that the smooth function
$\kappa\mapsto -\frac{n-2\kappa}{\sqrt{n\kappa(n-\kappa)}}$ has positive
derivative on $(0,n)$.
\end{proof}


\subsection{Classification of reducible symmetric equilibria}

The next theorem generalizes Theorem~\ref{cor.tech} and characterizes all
symmetric equilibria for the projected pure-influence model and its proof
can be found in Appendix~\ref{sec:proofs}.

\begin{theorem}[All equilibria for the projected pure-influence model]
\label{th_multi}
The matrix $Z^*$ is an equilibrium~\eqref{eq.equil} of the projected pure-influence model if and only if a permutation matrix $S$ exists such that:
\begin{enumerate}[label=(\roman*)]
\item $SZ^*S^{-1}=\diag(Z_1^*,\ldots,Z_s^*)$, $s\ge 1$, $Z_i^*={Z_i^*}^{\top}\in\R^{n_i\times n_i}$;
\item $Z_i^*=p_iVV^{\top}-q_iI_{n_i}$, where $p_i,q_i\ge 0$ and $V\in \nSt(n_i,k_i)$, $k_i<n_i$;
\item \label{s3}the sign $\ve=\sign(n_i-2k_i)\in\{-1,0,1\}$ is the same for all $i=1,\ldots,s$ such that $Z_i^*\ne \vect{0}_{n_i\times{n_i}}$ and
\item \label{s4}for each block $Z_i^*\ne \vect{0}_{n_i\times{n_i}}$ the coefficients $p_i,q_i$ have the form
\begin{equation}\label{eq.pq-general}
p_i=2\sqrt{\alpha^2+\beta_i},\quad q_i=\sqrt{\alpha^2+\beta_i}-\alpha,
\end{equation}
where
\begin{enumerate}
\item for $\ve\ne 0$, $\alpha$ and $\beta_i$ are determined from \label{a-s4}
\begin{equation}\label{eq.alpha-beta}
\begin{split}
&\alpha=\ve\Bigg(\sum_{i:Z_i\ne 0}\frac{4k_in_i(n_i-k_i)}{(n_i-2k_i)^2}\Bigg)^{-1/2}\\
&\beta_i=\alpha^2\frac{4n_ik_i-4k_i^2}{(n_i-2k_i)^2};
\end{split}
\end{equation}
\item for $\ve=0$, $\alpha=0$, for all $i$, and $\beta_i$ are chosen in
  such a way that $\sum_{i:Z_i\ne 0}\beta_in_i=1$.
\end{enumerate}
\end{enumerate}
\end{theorem}

\begin{remark}
Let $Z^*$ be a reducible equilibrium for the projected pure-influence model
such that $G(Z^*)$ is composed of $m$ (disconnected) subgraphs that satisfy
structural balance.
According to Definition~\ref{def:sbal}, $G(Z^*)$ \emph{does not} satisfy
structural balance since this definition requires $G(Z^*)$ to be complete.
\end{remark}


\subsection{Structural balance and equilibria}

We now characterize the equilibria corresponding to structural balance and
how they minimize the dissonance function.

\begin{corollary}[Balanced equilibria of the projected pure-influence model]\label{thm:bal-eq}
For the projected pure-influence
model~\eqref{def:projected-pure-influence}, let $Z^*\in\Qzu$ be an
equilibrium point with a single positive eigenvalue. Then,
  \begin{enumerate}
  \item\label{fact:equilibrium:2} $Z^*$ has the form
  \begin{equation} \label{pre_otr}
    Z^*=
    \left[
      \begin{array}{c|c}
        Z' & \vect{0}_{n_1\times{n-n_1}}\\
        \hline
        \vect{0}_{n-n_1\times{n_1}} & \vect{0}_{n-n_1\times{n-n_1}}
      \end{array}
      \right]
\end{equation}
with $n_1\leq n$ and
%
  \begin{equation} \label{eq:structural-eq0}
      Z'=\frac{1}{\sqrt{n_1(n_1-1)}} (s s^\top - I_{n_1}),
    \end{equation}
    for some $s\in\{-1,+1\}^{n_1}$; and thus, for any fixed $n_1$, there
    are only $2^{n_1-1}$ different equilibria (with a single positive
    eigenvalue),
  \item\label{fact:equilibrium:3}$G(Z')$ satisfies structural balance, with
    the binary vector $s$ characterizing the distribution of the
    individuals in the single faction or in the two factions, and
  \item \label{unss} if $G(Z^*)$ is a connected graph, then $G(Z^*)$
    satisfies structural balance, $Z^*$ is a global minimizer to the
    optimization problem:
    \begin{equation*}
      \begin{aligned}
        & \underset{Z\in\R^{n\times{n}}}{\textup{minimize}}
        & & \D(Z) \\
        & \textup{subject to}
        & &Z\in\Qzus
      \end{aligned}
    \end{equation*}
    and satisfies $\D(Z^*)=-\frac{n-2}{\sqrt{n(n-1)}}$.
  \end{enumerate}
\end{corollary}

\begin{proof}
Statement~\ref{fact:equilibrium:2} follows immediately from Theorem~\ref{cor.tech} and equation~\eqref{eq.k1-general}. Indeed, from Theorem~\ref{th_multi} we know that $Z'$
must be irreducible.
%
Regarding statement~\ref{fact:equilibrium:3}, observe that for any different $i$, $j$ and $k$,
\begin{equation*}
z'_{ij}z'_{jk}z'_{ki}=\frac{(s_is_j)(s_js_k)(s_ks_i)}{(n_1(n_1-1))^{3/2}}=\frac{1}{(n_1(n_1-1))^{3/2}}>0.
\end{equation*}
This inequality implies $\sign(z'_{ij})=\sign(z'_{jk}z'_{ki})$ and thus we
know that $Z'$ satisfies structural balance. It is immediate to see that
any $i$ and $j$ such that $s_i=s_j$ correspond to the same faction in the
network $G(Z')$.  This completes the proof for~\ref{fact:equilibrium:3}.

Regarding statement~\ref{unss}, we notice that the smooth function
$\eta\mapsto-\frac{\eta-2}{\sqrt{\eta(\eta-1)}}$ has negative derivative
for $\eta>3/2$. Then, if an equilibrium point with a single positive
eigenvalue of the form~\eqref{pre_otr} is a candidate solution to the shown
optimization problem, then it must be the case that $n_1=n$, i.e., the
graph associated with such equilibrium point is complete. Now, let us focus
on the evaluation of $\D$ on the equilibria of the projected pure-influence
model. First, let us have $k_1+\dots+k_s=k$ and $n_1+\dots+n_s=n$ for any
$s\geq 2$, where $k_i$ and $n_i$ are positive integers, and assume that
$k<n/2$ and $k_i<n_i/2$ for any $i\in\until{s}$.  Note that the function
$f(\xi)=\xi(1-\xi)/(1-2\xi)^2$ is convex on $(0,1/2)$. Therefore, Jensen's
inequality implies
\begin{multline*}
\frac{1}{n}\sum_{i=1}^s\frac{k_in_i(n_i-k_i)}{(n_i-2k_i)^2}=\sum_{i=1}^s\frac{n_i}{n}f\Big(\frac{k_i}{n_i}\Big)\geq f\Big(\sum_{i=1}^s\frac{k_i}{n}\Big)=f\Big(\frac{k}{n}\Big)=\frac{k(n-k)}{(n-2k)^2},
\end{multline*}
and, in turn,
\begin{equation}
  \label{l_ineq_aux}
  -\Bigg(\sum_{i=1}^s\frac{k_in_i(n_i-k_i)}{(n_i-2k_i)^2}\Bigg)^{-1/2}\geq -\frac{n-2k}{\sqrt{kn(n-k)}}.
\end{equation}
Now, let $Z^*$ and $Z^{**}$ be two equilibria with $k$ positive eigenvalues
being irreducible (as in Theorem~\ref{cor.tech}) and reducible with $s$
blocks (as in Theorem~\ref{th_multi}) respectively.  We immediately see
that, under our previous assumptions, the left hand side
of~\eqref{l_ineq_aux} corresponds to $\D(Z^{**})<0$ and the right hand side
corresponds to $\D(Z^*)<0$, so that $\D(Z^{**})\geq\D(Z^*)$. Thus, we only
need to investigate the minimum value of $\D$ in the set of irreducible
equilibria with $k<n/2$ positive eigenvalues in order to solve the
optimization problem, but the solution is already known
by~Theorem~\ref{cor.tech}\ref{la2} to be when $k=1$. This finishes the
proof.
\end{proof}

\begin{remark}
\label{remark_n-1}
Consider an equilibrium point $Z^*$ with one positive eigenvalue. 
Then, $-Z^*$ has one negative eigenvalue and $n-1$ positive eigenvalues, and does not correspond to structural balance. Note that all such $-Z^*$ correspond to critical points of $\D$ which are also isolated.
\end{remark}

\subsection{Examples of equilibria with two positive eigenvalues}
Let $Z^*$ be any equilibrium of the projected pure-influence model
parameterized by $\nSt(n,2)$ matrices, so that it has two positive
eigenvalues. Let us assume first that it is irreducible. Then, another
class of equilibria is found using the
parametrization~\eqref{eq.k2-general}. It can be easily shown that
\begin{equation*}
Z^*=\sqrt{\frac{2}{n(n-2)}}(\theta_{ij})_{i,j=1}^n,\quad \theta_{ij}=
\begin{cases}
0,\,i=j\\
\cos(\alpha_i-\alpha_j),i\ne j.
\end{cases}
\end{equation*}
Here the angles $\alpha_i$ should satisfy the relation~\eqref{eq.ngon}. Interestingly, many of such matrices do not correspond to structural balance. Consider, for example, the case where the unit vectors in~\eqref{eq.ngon} constitute a regular $n$-gon: $\alpha_i=\frac{\pi (i-1)}{n}$, $i=1,\ldots,n$. For any pair $i,j>i$ the entry $z_{ij}$ is negative if $(j-i)>n/2$, positive if $j-i<n/2$ and zero if $j-i=n/2$ (possible only for even $n$). If $n$ is odd, the graph is complete, otherwise, the pairs of nodes $(i,i+n/2)$ for $i=1,\ldots,n/2$ are not connected. For example, in the smallest dimension $n=3$, by setting $\alpha_1=0$, $\alpha_2=\pi/3$ and $2\pi/3$, we obtain the equilibrium
\begin{equation*}
Z^*=\frac{1}{\sqrt{6}}\begin{bmatrix}
0 & +1 & -1 \\
+1 & 0 & +1 \\
-1 & +1 & 0
\end{bmatrix}
\end{equation*}
which does not correspond to structural balance. Indeed, in the case where $n=3$ or $n\ge 5$,
the graph always contains imbalanced triads. For instance, for $n\ge 3$ being odd the nodes $i=1$, $j=(n-1)/2$ and $\ell=(n+3)/2$ always constitute such a triad: $z_{i\ell}<0$, whereas $z_{ij},z_{j\ell}\ge 0$. For an even number $n\ge 6$, one may take $i=1$, $j=n/2$, $\ell=n/2+2$.
In the case $n=4$, the equilibrium $Z^*$ corresponds to an incomplete cyclic graph such that $\D(Z^*)=0$:
\[
Z^*=\frac{1}{2\sqrt{2}}
\begin{bmatrix}
0 & \frac{1}{\sqrt{2}} & 0 & -\frac{1}{\sqrt{2}}\\
\frac{1}{\sqrt{2}} & 0 & \frac{1}{\sqrt{2}} & 0\\
0 & \frac{1}{\sqrt{2}} & 0 & \frac{1}{\sqrt{2}}\\
-\frac{1}{\sqrt{2}} & 0 & \frac{1}{\sqrt{2}} & 0
\end{bmatrix}.
\]

For the reducible matrix case, since $Z^*$ has two positive eigenvalues, $G(Z^*)$ contains two disconnected subgraphs that satisfy structural balance with possibly other isolated nodes.

\section{Convergence to balanced equilibria and stability analysis}
\label{sec:convergence-analysis}

We now provide convergence results for our models towards equilibria that
correspond to structural balance. We present a supporting lemma and then our main theorem.

\begin{lemma}\label{prop.aux}
  Assume that the solution of~\eqref{inf-dyn} satisfies
  $x_{i*}(t_0)=\vect{0}_{1\times{n}}$ at some $t_0\ge 0$, that is, in the
  graph $G(X(t_0))$ node $i$ does not communicate to any other node. Then,
  $x_{i*}(t)\equiv\vect{0}_{1\times{n}}$ for any $t\ge 0$.  The same holds
  for the solutions of~\eqref{def:projected-pure-influence}.
\end{lemma}
\begin{proof}
Since the right-hand sides of~\eqref{inf-dyn} and~\eqref{def:projected-pure-influence} are analytic, any solution is a real-analytic function of time. Assuming that $x_{ij}(t_0)=0$ for all $j$, one finds that $\dot x_{ij}(t_0)=0$. Differentiating~\eqref{inf-dyn}, it is easy to show that $\ddot x_{ij}(t_0)=0$, and so on, $x_{ij}^{(m)}(t_0)=0$ for any $m\ge 1$. In view of analyticity, one has $x_{ij}(t)\equiv 0$ for any $t$. Similarly, $z_{ij}(t_0)=0\,\forall j$ entails that $z_{ij}(t)\equiv 0$ for any solution of~\eqref{def:projected-pure-influence}.
\end{proof}

\begin{theorem}[Convergence results and dynamical properties]\label{th-mini}
  Consider the pure-influence model~\eqref{inf-dyn} with an initial
  condition $X(0)\in\Qzs$ and the projected pure-influence
  model~\eqref{def:projected-pure-influence} with initial condition
  $Z(0)=\frac{X(0)}{\Fnorm{X(0)}}$. Then,
  \begin{enumerate}
    \item the solution $Z(t)$ converges to a single critical point of the dissonance function $\D$;\label{lasym2}
    \item \label{lem-nondec}
    the number of negative eigenvalues of $Z(t)$ is non-decreasing.
    \setcounter{saveenum}{\value{enumi}}
  \end{enumerate}
  Moreover, if $X(0)$ has one positive eigenvalue, then
  \begin{enumerate}\setcounter{enumi}{\value{saveenum}}
  \item\label{one-thmm-mini} $\lim_{t\to+\infty}Z(t)=Z^*$, where $Z^*$ is
    as in~\eqref{eq:structural-eq0}, so that $G(Z(t))$ or one of its
    connected components (while the rest of nodes are isolated) reaches
    structural balance in finite time;
  \item\label{two-thmm-mini} $X(t)$ achieves the same sign structure as
    $Z^*$ in finite time;
  \item\label{three-thmm-mini2} nonzero entries of $X(t)$ diverge to
    infinity in finite time.
  \end{enumerate}
\end{theorem}
\begin{proof}
For convenience, throughout this proof, let us denote $W(t)=\frac{X(t)}{\Fnorm{X(t)}}$, i.e., $X(t)=\eta(t)W(t)$ with $\eta(t)$ evolving according to~\eqref{eq1o-1} and $W(t)$ evolving according to~\eqref{eq1o-2}. From the construction of the transcription of the pure-influence model in Theorem~\ref{eg-trans}, we have that $\eta(t)=\Fnorm{X(t)}$ and so $\eta(t)>0$ for all well-defined $t\geq 0$. Moreover, Lemma~\ref{thaux2} let us conclude that $W(t)=Z(\int_{0}^t\eta(s)ds)$ for all $t\geq 0$, and thus the solution $X(t)$ is well defined.

To prove~\ref{lasym2}, recall that~\eqref{def:projected-pure-influence} is a gradient flow dynamics of the analytic function $\D$, and the trajectory $Z(t)$ stays on a compact manifold and, in particular, is bounded. The classical result of \L ojasiewicz~\cite{PAA-RM-BA:05} implies convergence of the trajectory to a single fixed point.

To prove~\ref{lem-nondec}, we enumerate the eigenvalues of $Z(t)$ in the descending order $\la_1(t)\ge\la_2(t)\ldots\ge\la_n(t)$ and consider the corresponding orthonormal bases of eigenvectors $v_i(t)$. Since $Z_i(t)v_i(t)=\la_i(t)v_i(t)$ and $v_i(t)^\top v_i(t)=1$, we obtain $\dot{Z}v_i+Z\dot{v}_i=\dot{\lambda}_iv_i+\lambda_i\dot{v}_i$ and $\dot v_i(t)^\top v_i(t)=0$. Therefore,
$$
\dot{\lambda}_i=v_i^\top\dot{Z}v_i+v_i^{\top}Z\dot{v}_i=v_i^\top\dot{Z}v_i+\la_iv_i^{\top}\dot{v}_i=v_i^\top\dot{Z}v_i,
$$
entailing the following differential equation
\begin{align}
\label{eig-dif-o}
\dot{\lambda}_i=\lambda_i^2+\D(Z)\lambda_i-v_i^\top\diag(Z^2)v_i.
\end{align}
Notice that all diagonal entries of $\diag(Z^2)$ are nonnegative. Now, due
to Lemma~\ref{prop.aux}, if the $i$th row of $X$ was initially the zero
vector, then it will continue being the same for all times and also for
$Z$; and, moreover, $\diag(Z^2)_{ii}=0$ and there exists a zero eigenvalue
with its associated eigenvector having zero entries in all the positions of
the entries where $\diag(Z^2)$ are positive.  Then, it immediately follows
from~\eqref{eig-dif-o} that if $\lambda_{i}(0)=0$ due to $Z(0)$ having a
row being the zero vector $\vect{0}_{1\times n}$, then $\dot{\lambda_i}=0$.

Now, let $\mathcal{N}$ be the set of indices $i$ such that
$\diag(Z^2)_{ii}>0$.  Thus, for any $i\in\mathcal{N}$, if $\lambda_i$
crosses the real axis at time $t^*$, i.e., $\lambda(t^*)=0$, then
\begin{equation}
\label{equin}
\dot{\lambda}_i(t^*)=-(v_i(t^*))^\top\diag{(Z^2(t^*))}v_i(t^*)<0.
\end{equation}
Therefore, if $\la_i(t_0)\le 0$ for some $t_0\ge 0$, then $\la_i(t)\le 0$
for all $t\ge t_0$. This finishes the proof for~\ref{lem-nondec}.

%

Notice that since $\trace(Z(t))=0$ and $Z(t)=Z(t)^{\top}\ne \vect{0}_{n\times{n}}$, then $Z(t)$ has at least one positive eigenvalue. Then, equation~\eqref{equin} implies that
\begin{equation*}
\Lambda:=\setdef{Z\in\Qzus}{Z\text{ has only one positive eigenvalue}}
\end{equation*}
is forward invariant and, in particular, the limit $Z^*=\lim_{t\to\infty} Z(t)$ (existing in view of statement~\ref{lasym2}) belongs to $\Lambda$. Since $Z^*$ is a critical point of $\D$ (or, in view of Theorem~\ref{th-mini-gr}, the equilibrium of~\eqref{def:projected-pure-influence}), it has the structure described by
Corollary~\ref{thm:bal-eq}.

By continuity of the flow $Z(t)$, there is a finite time $\tau$ such that $G(Z(t))$ has the same sign structure as $G(Z^*)$ 
for all $t\geq\tau$.
This finishes the proof for~\ref{one-thmm-mini}.

Now we prove the last two statements of the theorem. Knowing the convergence result from~\ref{one-thmm-mini}, Lemma~\ref{thaux2} tells us that introducing the term $\eta$ as in the transcribed system~\eqref{eq1o-1} to the projected pure-influence model has the simple effect of altering the convergence rate properties for $Z(t)$. Therefore, there always exist a finite time $\tau^*\geq 0$ such that, for any $t\geq\tau^*$, $W(t)$ satisfies the sign properties of statement~\ref{one-thmm-mini} regarding structural balance. Moreover, the fact that $X(t)=\eta(t)W(t)$ and $\eta(t)\geq 0$ by construction, immediately implies~\ref{two-thmm-mini}. %
Now, let $g(t):=-\D(W(t))$, and notice that $g(t)$ is a strictly positive continuous function for all (well-defined) $t\geq\tau^*$. Now, from equation~\eqref{eq1o-2}, we have the system $\dot{\eta}(t)=g(t)\eta^2(t)$, with solution $\eta(t)=\frac{\eta(\tau)}{1-\eta(\tau)\int_\tau^tg(s)ds}$ for $t\geq\tau$. Then, since $\int_\tau^tg(s)ds$ is a monotonic strictly increasing function on $t\geq\tau$, we have that $\eta(t)\to+\infty$ as $t\to t^*$, where $t^*>\tau^*$ is some finite time such that $\int_\tau^{t^*}g(s)ds=\frac{1}{\eta(\tau)}$ (note that $t^*>\tau^*$ holds from the relationship $W(t)=Z(\int_0^{t}\eta(s)ds)$). Then, we conclude that the solution $\eta(t)$ and the entries of $X(t)$ diverge in some finite time $t^*$,  
which proves~\ref{three-thmm-mini2}.
\end{proof}

\begin{corollary}
  Consider the same conditions as in Theorem~\ref{th-mini}, i.e., the
  projected pure-influence model with initial condition $Z(0)\in\Qzus$
  having one positive eigenvalue. If
  $\D(Z(0))<-\frac{n-3}{\sqrt{(n-1)(n-2)}}$, then $G(Z(t))$ eventually
  reaches structural balance.
\end{corollary}


The previous theorem immediately implies that the set of irreducible
equilibria with a single positive eigenvalue is (locally) asymptotically
stable. We present further results on the stability of equilibria.


\begin{lemma}[Further results on stability of the equilibria]\label{l2}
  Consider a symmetric equilibrium point $Z^*$ for the projected
  pure-influence model~\eqref{def:projected-pure-influence}. Without loss
  of generality, assume that $Z^*$ has no row equal to the zero
  vector\footnote{If $Z^*$ had a row equal to the zero vector, then, in the
    lemma statement, we would replace $n$ by $n_1<n$, where $n_1$ is the
    number of rows of $Z^*$ that are not equal to the zero vector.}.
  If $\D(Z^*)\geq 0$, then $Z^*$ is an unstable equilibrium point and does
  not correspond to structural balance.
\end{lemma}
\begin{proof}
Write the analytic projected influence
system~\eqref{def:projected-pure-influence} as
$\dot{Z}=f(Z):=Z^2-\diag(Z^2)+\D(Z)Z$, thereby defining
$\map{f}{\R^{n\times{n}}}{\R^{n\times{n}}}$, and compute
\begin{align*}
  \frac{\partial f_{ij}(Z)}{\partial z_{ij}} &= {\D(Z)+\frac{\partial\D(Z)}{\partial z_{ij}}z_{ij},}\\
  \frac{\partial \D(Z^*)}{\partial z_{ij}}&= -3\sum\nolimits_{\substack{k=1\\k\neq i,j}}^n z^*_{ik}z^*_{kj}.
\end{align*}
Now, the Jacobian of $f$, denoted by $\jac{f}$, is a
$(n^2-n)\times(n^2-n)$ matrix (since we do not consider
self-appraisals). Let $\jac{f}(Z^*)$ be the Jacobian evaluated at $Z^*$ and
let $\{\lambda_i\}_{i=1}^{n^2-n}$ be the set of its eigenvalues. Then, we
compute
\begin{align*}
  \sum_{i=1}^{n^2-n}\lambda_{i}&=\trace(\jac{f}(Z^*))=\sum_{i=1}^n\sum\nolimits_{\substack{j=1\\j\neq i}}^n\frac{\partial f_{ij}(Z^*)}{\partial z_{ij}}\\
  &=(n^2-n)\D(Z^*)+3\D(Z^*)=(n^2-n+3)\D(Z^*).
\end{align*}
Since $n^2-n+3>0$ for $n\geq 3$, we draw the following conclusions for
$\D(Z^*)\geq 0$: (i) $\jac{f}(Z^*)$ contains at least one positive
eigenvalue and so the equilibrium point $Z^*$ is unstable; (ii) at least
one triad in $G(Z^*)$ is unbalanced and so $Z^*$ does not correspond to
structural balance.
\end{proof}

\section{Simulation results and conjectures}
\label{sec:simulations}

The generic convergence of trajectories to the minima of $\D$ (or,
equivalently, the convergence from almost all initial conditions) is an
open problem. However, we present strong numerical evidence that support
such claim.  We first remark that, from the proof of
Theorem~\ref{eg-trans}, the projected pure-influence
model~\eqref{def:projected-pure-influence} can be generalized over any
asymmetric matrix in $\Qzu$ by replacing $\D(Z)$ by $-\trace(Z^\top Z^2)$
and this is the model we will refer throughout this section.

A \emph{generic asymmetric initial condition} $X(0)$ for the pure-influence model~\eqref{inf-dyn}
is a matrix that is generated with each entry independently sampled from a uniform distribution with support $[-100,100]$, and its diagonal entries set to zero. A \emph{generic symmetric initial condition} is similarly constructed by only sampling the upper triangular entries of the matrix. 
For the projected pure-influence model, 
we say $Z(0)=\frac{X(0)}{\Fnorm{X(0)}}$ is a (non-)symmetric generic initial condition depending on how $X(0)$ was generated.
We immediately see from the proof of Theorem~\ref{th-mini}, 
that $Z(t)$ converges to social balance if and only if $X(t)$ converges to social balance. Indeed, given that $X(t)$ diverges at some finite time $\bar{t}$, we have $Z(\infty)=\frac{X(\bar{t}^-)}{\Fnorm{X(\bar{t}^-)}}$.

For a fixed network size $n$, we use a Monte Carlo
method~\cite{RT-GC-FD:05} to estimate the probability $p$ of the event
``under a generic asymmetric initial condition $Z(0)$, $Z(t)$ converges to
structural balance in finite time". We estimate $p$ by performing $N$
independent simulations (i.e., each simulation generates a new independent
initial condition) and obtaining the proportion $\hat{p}_N$, also known as
the empirical probability, of times that the simulation indeed had $Z(t)$
converging to structural balance in finite time.  For any accuracy
$1-\epsilon\in(0,1)$ and confidence level $1-\eta\in(0,1)$ we have that
$|\hat{p}_N-p|<\epsilon$ with probability greater than $1-\eta$ if the
Chernoff bound $N\geq\frac{1}{2\epsilon^2}\log\frac{2}{\eta}$ is
satisfied. For $\epsilon=\eta=0.01$, the bound is satisfied by $N=27000$.
We performed the $N=27000$ independent simulations with $n\in\{5,6\}$, and
found that $\hat{p}_N=1$. Our observations let us conclude that \emph{for
  generic asymmetric initial condition $Z(0)$ and $n\in\{5,6\}$, with
  $99\%$ confidence level, there is at least $0.99$ probability that $Z(t)$
  converges to structural balance in finite time.}

Similarly, we performed the same Monte Carlo analysis for generic symmetric
initial conditions with $n\in\{3,5,6,15\}$, and found for that
$\hat{p}_N=1$ for all $n$. Therefore, we conclude that \emph{for any
  symmetric generic initial condition $Z(0)$ and $n\in\{3,5,6,15\}$, with
  $99\%$ confidence level, there is at least $0.99$ probability that $Z(t)$
  converges to structural balance in finite time.}

We report three more observations and then state a resulting conjecture.
First, remarkably, we found that all of our simulations (for any type of
random initial condition) that converged to structural balance in finite
time, did it by converging to an equilibrium point having only one positive
eigenvalue inside the set of scale-symmetric matrices, which is a superset
of the set of symmetric matrices (see
Appendix~\ref{sec:AppScaSym}). Second, we did not perform experiments for
larger sizes of $n$ due to computational constraints.  Third,
unfortunately, for $n=3$, we did find randomly-generated asymmetric initial
conditions whose numerically-computed solutions do not converge to
structural balance.

\begin{conjecture}[Convergence from generic initial conditions]
Consider the pure-influence model~\eqref{inf-dyn} with some initial
condition $X(0)$, and the projected pure-influence
model~\eqref{def:projected-pure-influence} with initial condition
$Z(0)=\frac{X(0)}{\Fnorm{X(0)}}$.  Then,
  \begin{enumerate}
\item under generic asymmetric initial conditions, $\lim_{t\to+\infty}Z(t)=Z^*$ for a sufficiently large $n$,
\item under generic symmetric initial conditions, $\lim_{t\to+\infty}Z(t)=Z^*$ for any $n$,\label{opoo}
\end{enumerate}
where $Z^*$ is scale-symmetric (and particularly symmetric for~\ref{opoo})
corresponding to structural balance. Then, $Z(t)$ reaches structural
balance in finite time. Moreover, $X(t)$ reaches structural balance in
finite time with same sign structure as $Z^*$, and also diverges in finite
time.
\end{conjecture}

Similarly, we performed the same simulation analysis for the \Kulakowski et
al. model~\eqref{def:projected-pure-influenceK}, which converges to
structural balance if and only if the projected \Kulakowski
model~\eqref{inf-dynK} does. To generate a generic initial condition for
this system, we generated an $n\times{n}$ matrix with each entry
independently sampled from a uniform distribution with support
$[-100,100]$, and then divide it by its Frobenius norm. We performed
$N=27000$ independent simulations with $n\in\{5,6\}$, and found that
\emph{for  generic initial condition $Z(0)$ and $n=5$, only $16.94\%$
  converged to structural balance, and for $n=6$, only $11.50\%$ converged
  to structural balance. }

Also, for $n=3$, not all simulations converged to structural balance. We
remark that not all of the networks for which the system converged and did
not satisfy structural balance were complete, some of them were networks
with only self-loops, e.g., Figure~\ref{f:sim3}(a). Similarly, we performed
the same Monte Carlo analysis for symmetric initial conditions with
$n\in\{3,5,6,15\}$. Our results show that \emph{for symmetric generic
  initial condition, $Z(0)$ did not always converge to structural balance
  for $n=3$, but, for $n\in\{5,6,15\}$, with $99\%$ confidence level, there
  is at least $0.99$ probability that $Z(t)$ converges to structural
  balance in finite time.}

These Monte Carlo results are expected, since it has been formally proved
that the \Kulakowski et al. model converges to structural balance only
under generic symmetric initial conditions as
$n\to\infty$~\cite{SAM-JK-RDK-SHS:11} and negative results for asymmetric
conditions are given by~\cite{VAT-PVD-PDL:13}.


See Figure~\ref{f:sim1} for a comparison of trajectories of the
pure-influence model in both generic and symmetric generic initial
conditions. Figure~\ref{f:sim2} shows a comparison between our projected
pure-influence model, which does not consider self-appraisals, and the
projected influence model, which considers self-appraisals. Note how not
considering self-appraisals drastically change the convergence time as well
as the dynamic behavior of the interpersonal appraisals.

\begin{figure}[ht]
  \centering
  \subfloat[Projected pure-influence model~\eqref{def:projected-pure-influence} with generic asymmetric initial condition]{\label{f:11-a}\includegraphics[width=0.485\linewidth]{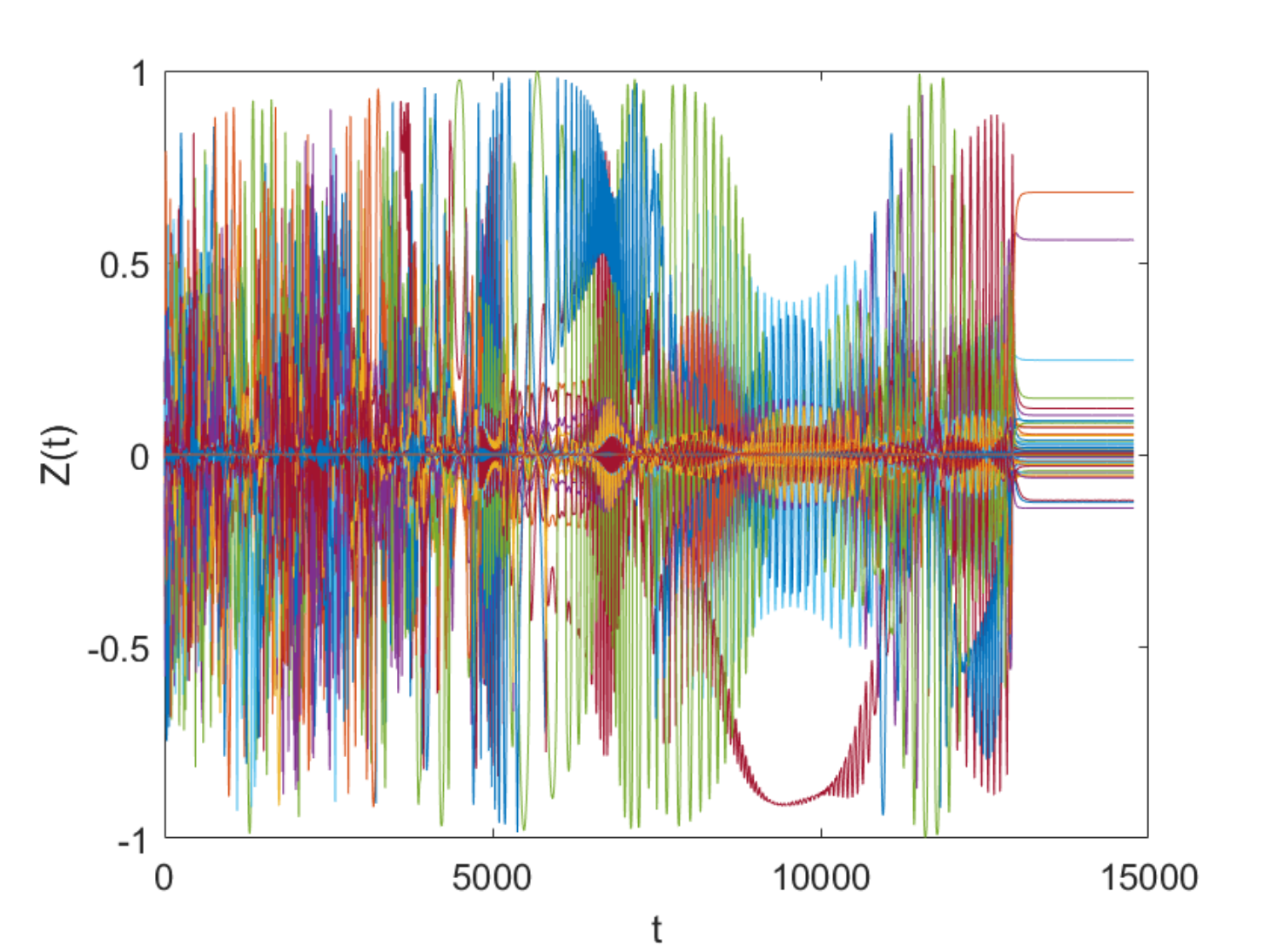}}
  \hfil
  \subfloat[Projected pure-influence model~\eqref{def:projected-pure-influence} with generic symmetric initial condition]{\label{f:14-a}\includegraphics[width=0.485\linewidth]{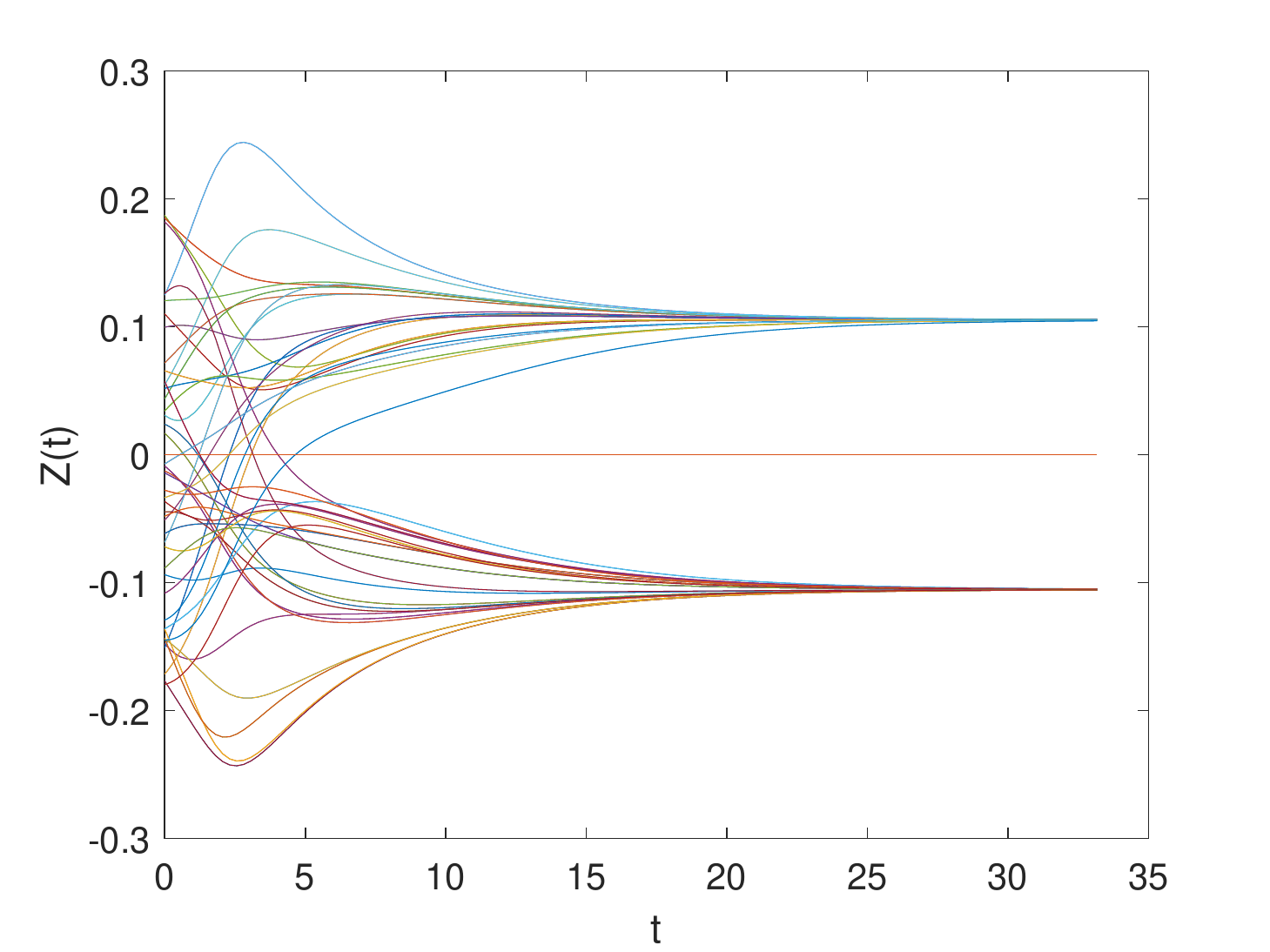}}
\caption{Convergence to structural balance for a network of size $n=10$. 
We plot the evolution of all the entries of $Z(t)$.
}
  \label{f:sim1}
\end{figure}

\begin{figure}[ht]
  \centering
  \subfloat[Projected influence model~\eqref{def:projected-pure-influenceK} with generic asymmetric initial condition]{\label{f:11-b}\includegraphics[width=0.485\linewidth]{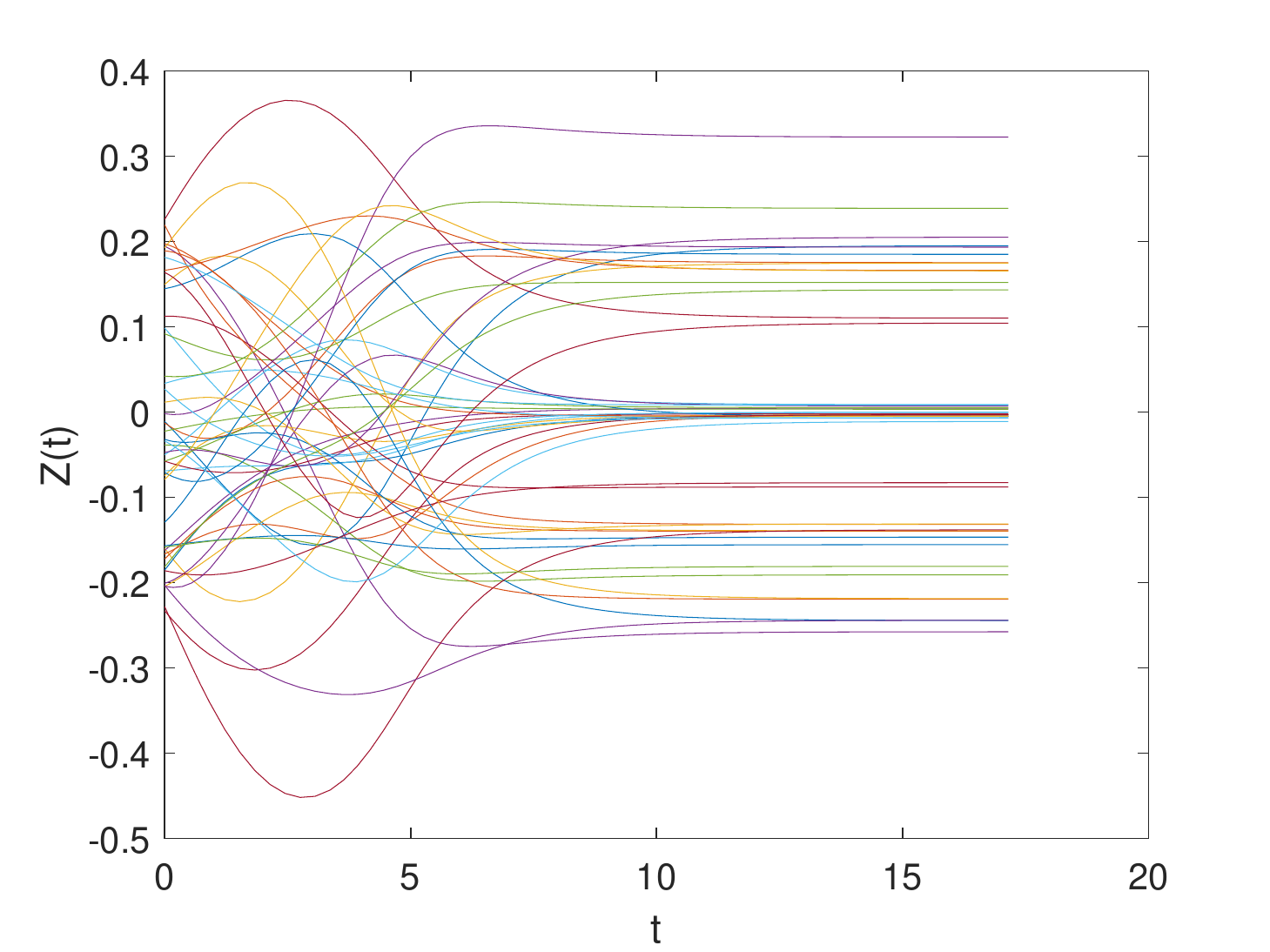}}
    \hfil
  \subfloat[Projected pure-influence model~\eqref{def:projected-pure-influence} with generic asymmetric initial condition]{\label{f:14-b}\includegraphics[width=0.485\linewidth]{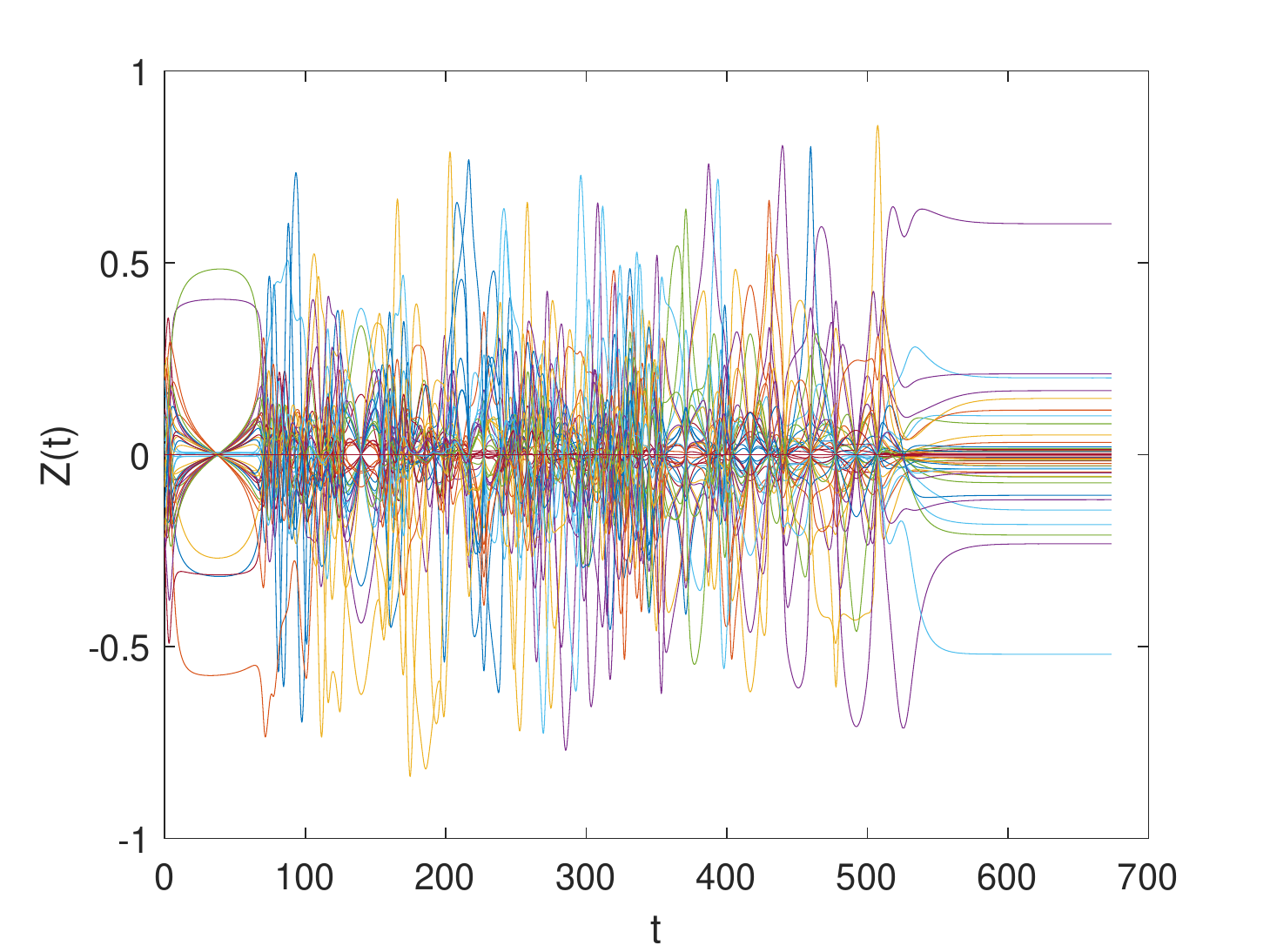}}
\caption{Convergence comparison for a network of size $n=7$ 
(a) with and (b) without the consideration of self-appraisals. 
We first generated an $n\times{n}$ random matrix $W$ with each entry independently sampled from a uniform distribution with support $[-100,100]$. Then, for (a), we normalize this matrix to have unit Frobenius norm and used it as the initial condition. For (b), we set the diagonal entries of $W$ to zero and then normalize it to have unit Frobenius norm and use it as the initial condition. In this example, (a) did not converged to structural balance, whereas (b) did. 
We plot the evolution of all the entries of the appraisal matrix.
}\label{f:sim2}
\end{figure}
\begin{figure}[ht]
  \centering
  \subfloat[Projected influence model~\eqref{def:projected-pure-influenceK} with generic asymmetric initial condition]{\label{f:11-c}\includegraphics[width=0.485\linewidth]{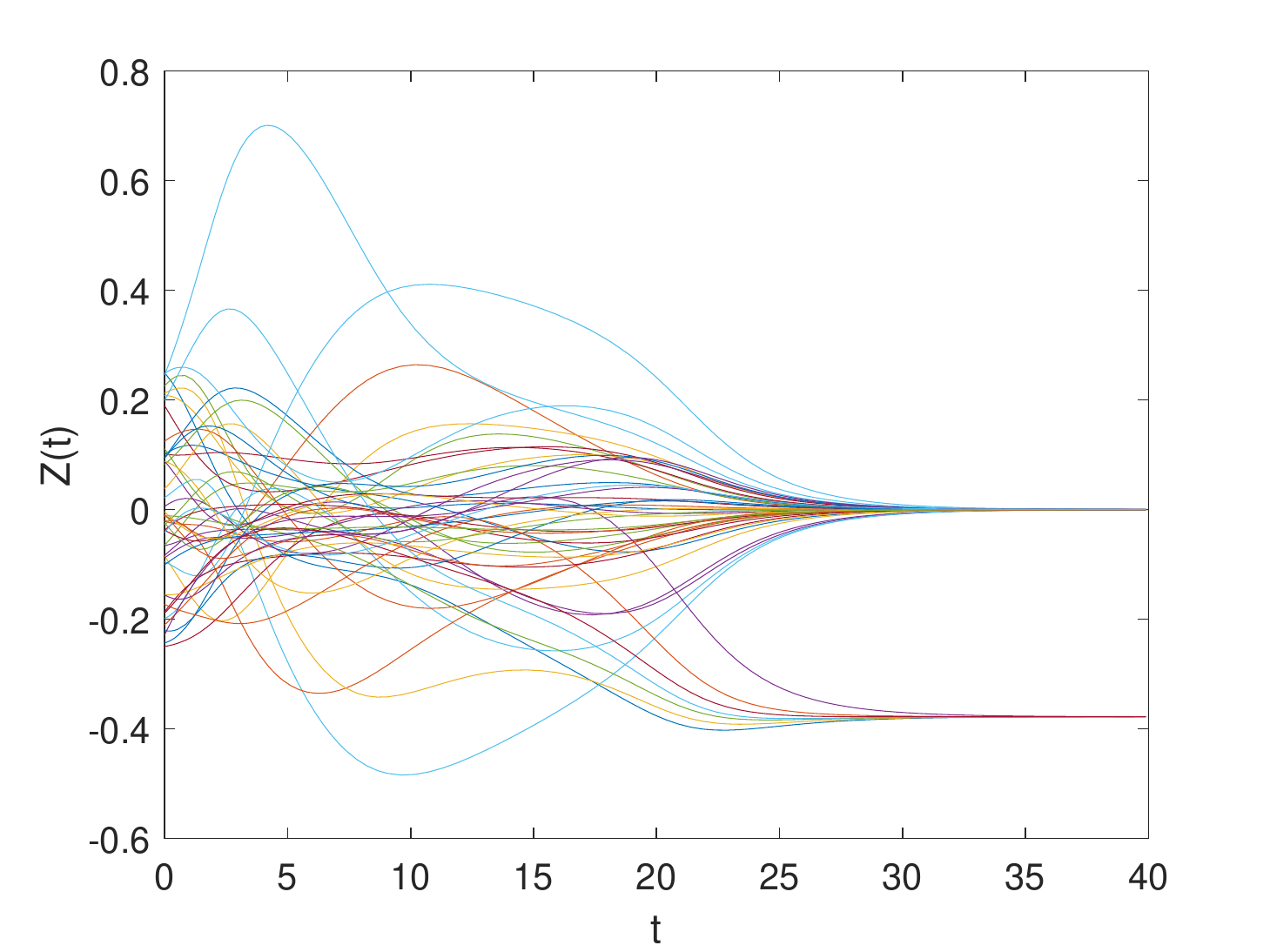}}
  \hfil
  \subfloat[Projected pure-influence model~\eqref{def:projected-pure-influence} with generic asymmetric initial condition]{\label{f:14-c}\includegraphics[width=0.485\linewidth]  {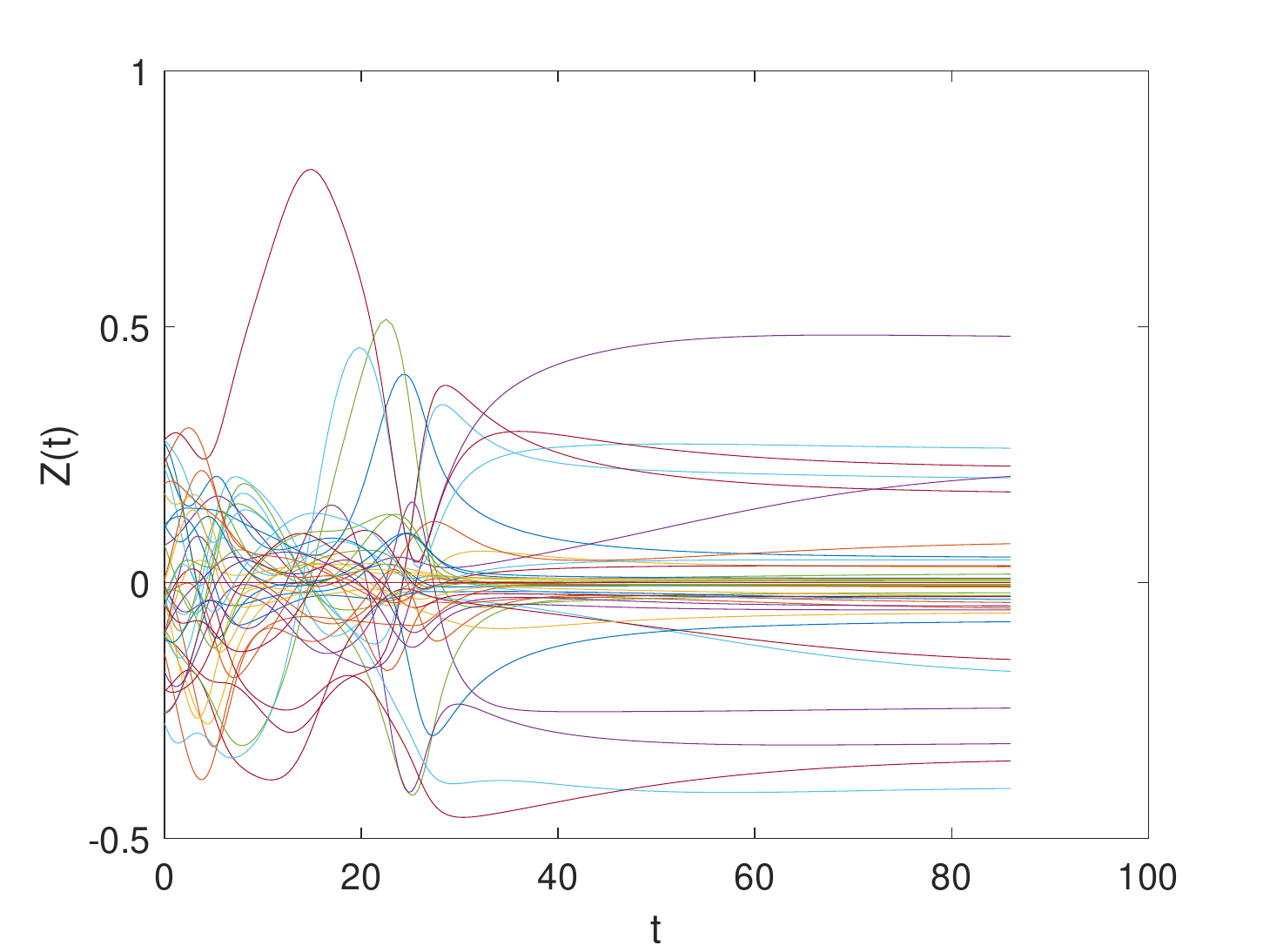}}
\caption{Convergence comparison for a network of size $n=7$ (a) with and (b) without the consideration of self-appraisals.
The setting is the same one as in Figure~\ref{f:sim2}, but with a different random initial condition. (a) converged to a network with only diagonal negative entries (all interpersonal appraisals go to zero), whereas (b) converged to structural balance.
}
  \label{f:sim3}
\end{figure}

\section{Conclusion}
\label{sec:conclusion}
We propose two new dynamic structural balance models that incorporates more psychologically plausible assumptions than previous models in the literature, based on a modification by a model proposed by \Kulakowski et al. We have established important convergence properties for these models and also that, most importantly, they correspond to gradient systems over an energy function that characterizes the violations of Heider's axioms for the symmetric case. We also expanded our results to a set of asymmetric matrices called scale-symmetric. Numerical results illustrates that, under generic initial conditions, our models converges to structural balance (for sufficiently large $n$) and thus have better convergence properties than the previous model by \Kulakowski et al. 

As future work, we propose to further study the general case of asymmetric
(and non-scale-symmetric) equilibria and the convergence properties of our
models under arbitrary initial conditions. For example, numerical
simulations of the projected pure-influence model from generic initial
conditions illustrate how this system features transient chaos before
converging towards an equilibrium.  A second future direction of work is to
find models with a more sociologically justified transient behavior from
generic initial conditions. Finally, another future direction is to study
the removal of the self-appraisals in other dynamical structural balance
models, like the homophily-based Traag el al.\ model~\cite{VAT-PVD-PDL:13}.

\section*{Acknowledgment}
We are grateful to Prof.\ John Gilbert, Prof.\ Ambuj Singh, and Dr.\ Saber
Jafarpour for insightful discussions.

\bibliographystyle{plainurl+isbn}
\bibliography{alias,Main,FB}

%
\section{Supporting results and proofs}
\label{sec:proofs}

\begin{lemma}
  \label{thaux2}
  Let $x(t)$ be the solution to $\dot{x}=f(x)$ from initial condition
  $x(0)$, with $f$ being a continuously differentiable vector field. Let
  $\eta$ be a positive continuous scalar function. Then, $y(t)$ is the
  solution to $\dot{y}=\eta(t)f(y)$ with initial condition $y(0)=x(0)$ if
  and only if $y(t)=x(\int_{0}^t\eta(s)ds)$.
\end{lemma}
\begin{proof}
  Consider the time transformation $\bar{t}(t)=\int_0^t\eta(s)ds$, which is
  well-defined since it is continuous and monotonically increasing on $t$
  (recall that $\alpha(s)>0$ for $s\in[0,t]$), with $\bar{t}=0$ if and only ifq
  $t=0$. Now, from the chain rule, it follows that
  \begin{equation*}
    \frac{dy}{dt}=\frac{dx(\bar{t})}{d\bar{t}}\frac{d\bar{t}}{dt}= f(y)\eta(t),\quad y(0)=x(0).
  \end{equation*}
  This finishes proof of the ``if'' part. The ``only if'' part follows from
  the uniqueness theorem.
\end{proof}


\begin{proof}[\textbf{Proof of Lemma~\ref{ch-gunit}}]
First, to prove that the set $\nSt(n,k)$, $k\leq n$ is a submanifold of the compact Stiefel manifold, define the smooth map $\map{\Phi}{\St(n,k)}{\R^n}$ by $X\mapsto (\norm{X_{i*}}_2^2,\dots,\norm{X_{n*}}_2^2)^\top$, where $X_{i*}$ is the $ith$ row of $X$. Then, we have that $\nSt(n,k)=\Phi^{-1}((k/n,\dots,k/n)^\top)$ and it is easy to prove the mapping $\Phi$ has constant rank $n$. Thus, we use the Constant-Rank Level Set theorem~\cite{JML:03} to conclude our claim. The properties of compactness and analyticity are immediate from the definition of the set $\nSt(n,k)$, $k\leq n$.

Now, notice that conditions~\eqref{eq.gu1} and~\eqref{eq.gu2} from Definition~\ref{def-unit} impose, in total, $\frac{k(k+1)}{2}+n$ constraints on $kn$ independent variables, however, these constraints are linearly dependent: one of them can be removed (for instance, if one requires condition~\ref{eq.gu1} from Definition~\ref{def-unit}, then suffices to constrain only sums of $n-1$ rows, whereas the remaining sum automatically equals $k/n$)). Whenever $k\le n$ and $n\ge 3$, one has $\frac{k(k+1)}{2}+n-1<kn$, which implies that the set $\nSt(n,k)$ has the dimension $(k-1)n+1-k(k+1)/2$.

Statements~\ref{st-1} and~\ref{st-2} are immediate. Now regarding~\ref{st-3}, it is obvious that each row has norm $\sqrt{k/n}$ if and only if $V$ can be written as~\eqref{eq.k2-general}.
Notice now the columns are unit vectors if and only if $\sum_{m=1}^n\cos^2\alpha_i=n/2=\sum_{m=1}^n\sin^2\alpha_i$, which in turn holds if and only if $\sum_m\cos 2\alpha_m=2\sum_m\cos^2\alpha_m-n=0$. Similarly, the columns are orthogonal if and only if $\sum_{m=1}^n\cos\alpha_i\sin\alpha_i=0=\frac{1}{2}\sum_m\sin 2\alpha_m$. These two constraints are equivalent to~\eqref{eq.ngon}.
\end{proof}

\begin{proof}[\textbf{Proof of Lemma~\ref{lem.tech2}}]
The case where $\alpha=\beta=0$ is trivial: $Z=0$ and it obviously can be decomposed as in~\eqref{eq.special1} with $p=q=0$. Notice that every eigenvalue of $Z=Z^{\top}$ corresponds to the eigenvalue $\la^2-2\alpha\la$ of $Z^2-2\alpha Z$, and hence $\la^2-2\alpha-\beta=0$. Therefore, $\alpha^2+\beta\ge 0$ (otherwise, eigenvalues of $Z$ would be complex).
Furthermore, $\alpha^2+\beta\ne 0$ (otherwise, $\la=\alpha$ would be the
only eigenvalue of $Z$ of multiplicity $n$, and one would have $\trace(Z)=\alpha n$, entailing that $\alpha=\beta=0$). Denoting
$\Delta=\sqrt{\alpha^2+\beta}$, the matrix $Z$ has two different
eigenvalues $\alpha+\Delta$ and $\alpha-\Delta$, denote their
multiplicities by $k$ and $n-k$. Then
$(\alpha+\Delta)k+(\alpha-\Delta)(n-k)=0$.  Denoting $q=\Delta-\alpha$ and
$p=2\Delta>0$, one has $(p-q)k-q(n-k)=0$ or, equivalently, $pk=qn$ thus, $q>0$.

Consider the orthonormal eigenvectors $v_1,\ldots,v_k$, corresponding to
the eigenvalue $p-q=\alpha+\Delta$ and orthonormal eigenvectors
$w_{1},\ldots,w_{n-k}$, corresponding to $-q=\alpha-\Delta$. The sequence
$v_1,\ldots,v_k$, $w_1$, $\ldots$, $w_{n-k}$ constitutes an orthonormal basis of
eigenvectors for the operator $Z$.  Stacking the columns $v_i$ and $w_i$,
one obtains $n\times k$ and $n\times (n-k)$ matrices $V=(v_1,\ldots,v_k)$,
$W=(w_1,\ldots,w_{n-k})$. The matrix $[V,W]$ is orthogonal and diagonalizes
$Z$:
$Z[V,W]=
[V,W]\begin{bmatrix}
(p-q) & 0\\
0 & -q
\end{bmatrix}$ and thus $Z=(p-q)VV^{\top}-qWW^{\top}$.
Since $VV^{\top}+WW^{\top}=I_n$, $Z$ is decomposed as~\eqref{eq.special1}. It remains to notice that $V^{\top}V=I_k$ by
definition of the orthonormal basis and $\diag VV^{\top}=(q/p)I_n=(k/n)I_n$ since $\diag Z=0$. To finish the proof, notice that $p-2q=2\alpha$ and $\beta=\Delta^2-\alpha^2=(\Delta-\alpha)(\Delta+\alpha)=q(p-q)$.
%
%
\end{proof}

\begin{proof}[\textbf{Proof of Corollary~\ref{cor.diag}}]
Denoting $f(z)=z^2-2\alpha z$, $z\in\mathbb{C}$, it suffices to show that if $f(Z)=\diag(\beta_1I_{n_1},\ldots,\beta_sI_{n_s})$, then $Z=\diag(Z_1,\ldots,Z_s)$, where the diagonal blocks obey the equations $f(Z_i)=\beta_iI_{n_i}$. This statement will be proved for any analytic function $f(z)$. It is well known that the spectrum of $f(Z)$ consists of all points $f(\la)$, where $\la$ is an eigenvalue of $Z$.
Consider the set of eigenvalues of $Z$ that belong to $f^{-1}(\beta_i)$ and let ${\cal X}_i$ be the sum of corresponding eigenspaces. Then ${\cal X}_i$ is invariant under the operator $Z$, and $\R^n=\oplus_{i=1}^s{\cal X}_i$ (the sum is orthogonal). Also, $f(Z)x=\beta_i x$ for any $x\in{\cal X}_i$. For any basis vector $e_r=(0,\ldots,1,\ldots,0)^{\top}$ consider the decomposition $e_r=\oplus_{i=1}^se_r^i$, $e_r^i\in {\cal X}_i$. Then $Ze_r=\oplus_{i=1}^sZe_r^i$, $Ze_r^i\in{\cal X}_i$ and $f(Z)e_r=\oplus_{i=1}^sf(Z)e_r^i=\oplus_{i=1}^s\beta_ie_r^i$. Suppose that $1\le r\le n_1$. Then $f(Z)e_r=\beta_1e_r$. Since $\beta_1,\ldots,\beta_s$ are pairwise different, we have $e_r=e_r^1$ and $e_r^2=\ldots=e_r^s=0$. Similarly, for $n_1+n_2+\ldots+n_{j-1}+1\le r\le n_1+n_2+\ldots+n_{j-1}+n_{j}$ one has $e_r=e_r^j$ ($j=2,\ldots,s$).

In other words, each ${\cal X}_i$ contains $n_i$ basis vectors $e_r$, where $n_1+n_2+\ldots+n_{i-1}+1\le r\le n_1+n_2+\ldots+n_{i-1}+n_{i}$ and thus $\dim{\cal X}_i\ge n_i$. Recalling that $n_1+\ldots+n_s=n$,
one shows that $\dim{\cal X}_i=n_i\,\forall i$ and thus ${\cal X}_i$ is spanned by the corresponding basis vectors. Since ${\cal X}_i$ is invariant under $Z$, $Z=\diag(Z_1,\ldots,Z_s)$, where the block $Z_i$ has dimension $n_i\times n_i$. Obviously, $f(Z_i)=\beta_i\diag I_{n_i}$. The statement of Corollary is now immediate from Lemma~\ref{lem.tech2}.
\end{proof}


\begin{proof}[\textbf{Proof of Theorem~\ref{th_multi}}]
We prove the necessity first. Denote $2\alpha=-\D(Z)$. By assumption, $Z^2-2\alpha Z$ is diagonal. Statements (i) and (ii) follow from Corollary~\ref{cor.diag}, entailing also that
$p_i,q_i$ can be represented as~\eqref{eq.pq-general} with some $\beta_i$.
Since $Z_i^2=2\alpha Z_i+\beta_i I_{n_i}$ and $\diag Z_i=0$, one has $\trace Z_i^2=\beta_in_i$, therefore
\begin{equation}\label{eq.aux1}
\sum_{i=1}^s\beta_in_i=\trace (Z^2)=1.
\end{equation}
Recall also that for each $i$ one has $p_ik_i=q_in_i$ or, equivalently,
\[
\frac{2k_i}{n_i}=\frac{\sqrt{\alpha^2+\beta_i}-\alpha}{\sqrt{\alpha^2+\beta_i}}=1-\frac{\alpha}{\sqrt{\alpha^2+\beta_i}}\quad\forall i:p_i,q_i\ne 0.
\]
(if $\alpha=0$, one always has $p_i,q_i\ne 0$, otherwise it is possible that $\beta_i=0$ and then $Z_i=0$). This implies
condition 3 ($\ve=\sign\alpha$) and allows to determine $\alpha,\beta_i$. In the case where $\ve\ne 0$ notice
that $n_i-2k_i\ne 0$ for any $i$ such that $Z_i\ne 0$. Thus
\[
\frac{\beta_i+\alpha^2}{\alpha^2}=\frac{n_i^2}{(n_i-2k_i)^2}\Longleftrightarrow \beta_i=\alpha^2\frac{4n_ik_i-4k_i^2}{(n_i-2k_i)^2}.
\]
In view of~\eqref{eq.aux1}, one obtains that
\[
\alpha=\ve\left(\sum_{i:Z_i\ne \vect{0}_{n_i\times{n_i}}}\frac{4k_in_i(n_i-k_i)}{(n_i-2k_i)^2}\right)^{-1/2},
\]
which entails~\eqref{eq.alpha-beta}. In the case of $\alpha=0$, one has $p_i=2\sqrt{\beta_i},q_i=\sqrt{\beta_i}$ for any $i$,
and~\eqref{eq.aux1} implies that $\sum_i q_i^2n_i=1$. This finishes the proof of statement~\ref{s4}.

The proof of sufficiency is similar. For any $i$ such that $Z_i\ne 0$, the coefficients $p_i,q_i$ have the form~\eqref{eq.pq-general} (if $\ve\ne 0$, this is implied by~\ref{a-s4}, otherwise we choose $\alpha=0$ and $\beta_i=q_i^2=p_i^2/4$). Therefore, we have $Z_i^2-2\alpha Z_i=\beta_i Z_i$ and, in particular,
$Z^2-2\alpha Z$ is diagonal. A straightforward computation shows that $p_ik_i=q_in_i$ and thus $\diag Z_i=0\,\forall i$,
in particular, $\diag Z=0$. Also, $\diag Z_i^2=\beta_in_i$, and statement~\ref{s4} now implies that $\trace Z^2=1$. It remains
to notice that $Z_i^3=2\alpha Z_i^2+\beta_iZ_i$, and hence $\trace(Z_i^3)=2\alpha\beta_in_i$. Hence, $\D(Z)=-\trace(Z^3)=-2\alpha$ and
$Z^2+\D(Z)Z$ is a diagonal matrix. This finishes that $Z$ is an equilibrium~\eqref{eq.equil}.
%
%
\end{proof}

\section{Scale-symmetric matrices}
\label{sec:AppScaSym}

We now generalize our results for symmetric appraisal networks to a class
of asymmetric matrices. We define the sets of \emph{scale-symmetric}
matrices
\begin{align*}
  \Psa &= \setdef{A\in \Qz}{\text{there exists }\gamma\succ\vect{0}_n \text{ such that }\\
  &A\diag(\gamma)=(A\diag(\gamma))^\top},\\
  \Ps &= \Qzu \intersection \Psa.
\end{align*}
Note that $\Ps \supset \Qzus$ and
$$
\begin{gathered}
\Ps=\bigcup_{\gamma\succ\vect{0}_n}{\Ps}(\gamma),\\
{\Ps}(\gamma)=\setdef{A\in\Qzu}{A\diag(\gamma)=(A\diag(\gamma))^\top}.
\end{gathered}
$$.


\begin{lemma}\label{thaux1}
Consider any $\gamma\succ\vect{0}_n$ and some matrix $A\in\R^{n\times{n}}$ such that $A\diag(\gamma)=\diag(\gamma)A^\top$. Then,
\begin{enumerate}[label=(\roman*)]
\item $A$ has real eigenvalues and it is diagonalizable, \label{auxa1}
\item $\trace(A^2)=0$ if and only if $A=0$. \label{auxa2}
\end{enumerate}
\end{lemma}
\begin{proof}
Since $A\diag(\gamma)$ is symmetric, then
$A'=\diag(\gamma)^{-1/2}A\diag(\gamma)^{1/2}$ is also symmetric and thus
has real eigenvalues and its eigenvectors form an orthogonal basis. Now,
let $(\lambda,v)$ be an eigenpair for $A'$. Then, by defining
$u=\diag(\gamma)^{1/2}v$, we observe that $Au=\lambda u$, and so
$(\lambda,\diag(\gamma)v)$ is an eigenpair for $A$. Hence the eigenvectors
of $A$ form a basis, and thus $A$ is diagonizable. This proves~\ref{auxa1}.

Observe that $A=\diag(\gamma)A^\top\diag(\gamma)^{-1}$. Then, $\trace(A^2)=\trace(A\diag(\gamma)A^\top\diag(\gamma)^{-1})$. From simple algebraic operations, it can be found that $\trace(A^2)=\sum_{i=1}^n\sum_{j=1}^n\frac{\gamma_j}{\gamma_i}a^2_{ij}$. Since $\frac{\gamma_i}{\gamma_j}>0$, $\trace(A^2)=0$ if and only if $A=0$. This proves~\ref{auxa2}.
\end{proof}

In view of Lemma~\ref{thaux1}, a matrix $A$ is scale-symmetric if and only
if $A = D^{-1}A_s D$, where $D > 0$ is a positive diagonal matrix (in
Lemma~\ref{thaux1}, $D = \diag(\gamma^{-1/2})$ for some
$\gamma\succ\vect{0}_n$) and $A_s$ a symmetric matrix.

Recall the invariance property of the pure-influence model~\eqref{inf-dyn}:
if $X(0)=X(0)^{\top}$, then $X(t)=X(t)^{\top}$ for all $t>0$.  We are now
ready to provide a more general version of this property: If $D>0$ is a
diagonal matrix and $X(t)$ is a solution, then $DX(t)D^{-1}$ is also a
solution. For this reason, if $X(0)=DX_s(0)D^{-1}$ is a scale-symmetric
matrix with some $X_s(0)=X_s(0)^{\top}$, then the solution
$X(t)=DX_s(t)D^{-1}$ is scale-symmetric. A similar result holds for the
projected pure-influence
model~\eqref{def:projected-pure-influence}. Indeed, all of the theoretical
results obtained in this paper for symmetric appraisal matrices, can be
generalized to scale-symmetric appraisal matrices. For example, if
$X(0)\in\Psa$ ($Z(0)\in\Ps$) then $t \mapsto \D(X(t))$ ($t \mapsto
\D(Z(t))$) is monotonically nondecreasing in $\Psa$ ($\Ps$).

\end{document}